\documentclass[a4paper,12pt]{article} 

\usepackage{amsmath,amsfonts,amsthm}
\usepackage[utf8]{inputenc}
\usepackage{algorithm2e}
\usepackage{graphicx}   
\usepackage[english]{my-shortcuts}
\usepackage[toc,page]{appendix}
\usepackage{multirow}
\begin{document}
\title{Simple models for multivariate regular variations and the H\"usler-Reiss Pareto distribution}
\author{Zhen Wai Olivier Ho
\footnote{Univ. Bourgogne Franche-Comt\'e, Laboratoire de Math\'ematiques de Besan\c{c}on, UMR 6623, 16 route de Gray, 25030 Besancon, France. E-mail: zhen\_wai\_olivier.ho@univ-fcomte.fr} 
\and Cl\'ement Dombry 
\footnote{Univ. Bourgogne Franche-Comt\'e, Laboratoire de Math\'ematiques de Besan\c{c}on, UMR 6623, 16 route de Gray, 25030 Besancon, France. E-mail: clement.dombry@univ-fcomte.fr} 
}

\date{}

\maketitle

\begin{abstract} We revisit multivariate extreme value theory modeling by emphasizing multivariate regular variations and the multivariate Breiman Lemma. This allows us to recover in a simple framework the most popular multivariate extreme value distributions, such as the logistic, negative logistic, Dirichlet, extremal-$t$ and H\"usler-Reiss models. In a second part of the paper, we focus on the H\"usler-Reiss Pareto model and its surprising  exponential family property. After a thorough study of this exponential family structure, we focus on maximum likelihood estimation. We also consider the generalized H\"usler-Reiss Pareto model with different tail indices and  a likelihood ratio test for discriminating constant tail index  versus varying tail indices.
\end{abstract}

\section{Introduction}
Regular variation is a fundamental notion in extreme value theory that was widely popularized  by Resnick  \cite{R87}. As a simple illustration of the importance a regular variation in univariate extreme value theory, consider an independent and identically distributed (i.i.d.) sequence $X,X_1,X_2,\ldots$ of positive random variables with cumulative distribution $F$. For $n\geq 1$, let $a_n=F^\leftarrow(1-1/n)$ be the quantile of order $1-1/n$ of $F$. Then the following statements are equivalent:
\begin{itemize}
\item[i)] the tail function $1- F$ is regularly varying at infinity with index $-\alpha<0$, i.e.
\[
\lim_{u\to\infty} \frac{1-F(ux)}{1-F(u)}=x^{-\alpha},\quad x>0;
\]
\item[ii)] the rescaled maximum $a_n^{-1}\max(X_1,\ldots,X_n)$  converge in distribution as $n\to\infty$ to a standard $\alpha$-Fr\'echet distribution, i.e.
\[
\lim_{n\to\infty}\mathbb{P}\left[a_n^{-1}\max(X_1,\ldots,X_n)\leq  x \right] =\exp(-x^{-\alpha}),\quad x>0;
\]
\item[iii)] the rescaled exceedance $u^{-1}X$ of $X$ given $X>u$ converge in distribution as $u\to\infty$ to a standard $\alpha$-Pareto distribution; i.e.
\[
\lim_{u\to\infty}\mathbb{P}\left[ u^{-1}X>x\mid X>u\right]=x^{-\alpha},\quad x>1. 
\]
\item[iv)] the sample point process $\{X_i/a_n,\ 1\leq i\leq n\}$ converge to a Poisson point process on $(0,\infty)$ with intensity $\alpha x^{-\alpha-1}\mathrm{d}x$.
\end{itemize}
The equivalence i)-ii) dates back to Gnedenko \cite{G43}, the equivalence ii)-iii) is due to Balkema and de Haan \cite{BdH74} and the equivalence i)-iv) can be found in Resnick \cite{R87}. As will be reviewed in Section \ref{subsec:multivariate-RV}, a similar result holds in the multivariate setting and multivariate regular variations is crucial in multivariate extreme value theory.

Historically, multivariate extreme value theory has been developped by considerations on the asymptotic behaviour of i.i.d. random vectors.  Key early contributions are the papers by Tiago de Oliveira
\cite{TdO58}, Sibuya \cite{S60}, de Haan and Resnick \cite{dHR77}, Deheuvels \cite{D78}. The general structure of multivariate extreme value distribution has been characterized by de Haan and Resnick \cite{dHR77} in terms of the so-called spectral representation. Domain of attractions have been characterized by Deheuvels \cite{D78} that pointed out the convergence of the dependence structure to an extreme value copula. Since then a rich literature has emerged on modeling or statistical aspects of the theory, of which a nice recent review from the copula viewpoint is provided by  Gudendorf and Segers \cite{GS10}. 

More recent developments focus on exceedances over high threshold in a multivariate setting and the so called multivariate generalized Pareto distributions. Seminal papers in that direction are Coles and Tawn \cite{CT91} and Rootzen and Tajvidi \cite{RT06}. Further recent development on modeling and statistical aspects include Rootze\'en et al. \cite{RSW17} and Kiriliouk et al. \cite{KRSW16}.

\smallskip
In this framework, the motivations of the present paper are twofold. In a first part corresponding to Section 2, we revisit  multivariate extreme value theory models and put the emphasis on regular variations and the limiting homogeneous measure. More precisely, a multivariate extension of the celebrated Breiman Lemma due to Davis and Mikosch \cite{DM08} allows us to construct a regularly varying random vectors as a product of a heavy tailed random variable (thought as a radial component) and a sufficiently integrable random vector (thought as a spectral component). The limiting homogeneous measure is easily characterized and, for specific choice of the spectral component, we recover standard parametric models from multivariate extreme value theory such as the H\"usler-Reiss \cite{HR89}, extremal-t \cite{NJL09}, logistic, negative logistic  or Dirichlet models \cite{CT91}. We believe putting the emphasis on the exponent measure is important since it is the fundamental notion that unifies maxima, exceedances or point processes approaches in extreme value theory.  On the other hand, from the copula point of view, the multivariate Breiman Lemma provides a general framework for deriving extreme value copula models closely related to the results by Nikoloulopoulas et al. \cite{NJL09} or Belzile and  Ne\v slehov\'a \cite{BN17}. 

The second part of the paper corresponds to Sections 3 and 4 and proposes a thorough study of the so-called Hüsler-Reiss Pareto model, that is the exceedance Pareto model associated with the max-stable Hüsler-Reiss model \cite{HR89}. The exceedances of the related  Brown-Resnick spatial model were considered recently by Wadsworth and Tawn \cite{WT14} who proposed inference via censored maximum likelihood, see also Kiriliouk et al. \cite{KRSW16}. Here, we focus on the finite-dimensional multivariate Husler-Reiss Pareto model and notice that it has a simple exponential family structure (see Bandorf-Nielsen \cite{BN17}), that seems to have been overlooked in the literature. We propose in Section 3 an extensive study of this exponential family structure and consider also  maximum likelihood inference as well as perfect simulation.  We extend these results in Section 4 where we consider the non-standard Husler-Reiss Pareto model that incorporates different tail parameters for the different margins. Maximum likelihood estimators are shown again to be asymptotically normal and an alternating optimization procedure is considered. To conclude, we propose a maximum likelihood ratio test for testing the equality of the different marginal tail parameters.

\medskip
\textbf{Vector notations:} we denote by $\|\cdot\|_\infty$ the max-norm on $\mathbb{R}^d$ and by $\|\cdot\|$ an arbitrary norm, $1_d=(1,\ldots,1)$ is the vecor with all components equal to $1$. Operations on vector are usually meant componentwise. The componentwise maximum of vector is denoted  $\max(x_1,x_2)=x_1\vee x_2$, the comparison of vectors $x_1\leq x_2$ is  meant componentwise so that $x_1\not\leq x_2$ means that some components of $x_1$ is larger than the corresponding component of $x_2$. For $x\in[0,\infty)^d$, we note $[0,x]$ the cube $[0,x_1]\times\cdots\times[0,x_d]$ and $[0,x]^c=[0,\infty)^d\setminus[0,x]$.

\section{A simple model for multivariate regular variation}

\subsection{Preliminaries on multivariate regular variation}\label{subsec:multivariate-RV}
Following Hult and Lindskog \cite{hult:lindskog:2006}, we define multivariate regular variation in terms of $M_0$-convergence in $\mathbb{R}^d$ rather than vague convergence in $[-\infty,\infty]^d\setminus\{0\}$. This is completely equivalent in the multivariate setting but $M_0$-convergence can be more easily generalized to a metric space.

Consider the space $M_0(\mathbb{R}^d)$ of Borel measures $\mu$ on $\mathbb{R}^d\setminus\{0\}$ that assigns finite mass on sets bounded away from $0$, that is $\mu(\mathbb{R}^d\setminus O)$ is finite for all $O$ open neighborhood of $0$. A sequence $\mu_n\in M_0(\mathbb{R}^d)$ is said to converge to $\mu\in M_0(\mathbb{R}^d) $, noted $\mu_n\stackrel{M_0}\rightarrow \mu$, if $\int f\mathrm{d}\mu_n\to \int f\mathrm{d}\mu$
for all bounded continuous function $f$ that vanishes on a neighborhood of $0$. 

A random vector $X$ on $\mathbb{R}^d$ is called regularly varying with sequence $a_n\to+\infty$ if 
\[
n\mathbb{P}(X/a_n\in\cdot) \stackrel{M_0}\rightarrow \Lambda \quad \mbox{as $n\to\infty$}
\] 
with a non-zero limit measure $\Lambda\in M_0(\mathbb{R}^d)$. Necessarily, there exists $\alpha>0$, called the tail index of $X$, such that the limit measure is homogeneous of order $\alpha$, i.e.
\[
\Lambda(uA)=u^{-\alpha}\Lambda(A) \quad u>0, A\subset \mathbb{R}^d\setminus\{0\} \mbox{ Borel}.
\]
Furthermore, the sequence $(a_n)$ is  regularly varying at infinity with index $1/\alpha$ and a possible choice for the normalizing sequence $a_n$ is 
\begin{equation}\label{eq:def-an}
a_n=\inf\{x>0;\ \mathbb{P}(\|X\|_\infty\leq x)\geq 1-1/n \},\quad n\geq 1.
\end{equation}
Due to its importance in multivariate extreme value theory, we emphasize here the case of random vectors with non negative components and regular variations on $[0,\infty)^d$. In this simple case, regular variation can be characterized by the convergence of the tail function, see Hult and Lindskog \cite{hult:lindskog:2006}: we have equivalence 
\begin{itemize}
\item[i)] the random variable $X$ is regularly varying on $[0,\infty)^d$ with limit measure $\Lambda$, that is
\[
n\mathbb{P}(a_n^{-1}X\in\cdot)\stackrel{M_0}\longrightarrow \Lambda(\cdot),\quad  \mbox{ as } n\to\infty;
\]
\item[i')] the tail function $1-F(x)$ is regularly varying with  limit function $V(x)=\Lambda([0,x]^c)$, i.e.
\[
\lim_{u\to+\infty} \frac{1-F(ux)}{1-F(u1_d)}=V(x),\quad x\in [0,\infty)^d\setminus\{0\};
\]
\end{itemize}

Paralleling the univariate extreme value theory and the equivalence i)-iv) mentioned in the introduction, we consider a sequence $X,X_1,X_2,\ldots$ of non negative random vectors with cumulative distribution $F$ on $[0,\infty)^d$ and we assume for convenience $\mathbb{P}(X=0)=0$.  The following statements are known to be equivalent, see e.g. the monograph by Resnick \cite{R87} or Coles \cite{Coles01}:
\begin{itemize}
\item[i)] (regular variation) the random variable $X$ is regularly varying on $[0,\infty)^d$ with $\alpha$-homogeneous limit measure $\Lambda$;
\item[ii)] (componentwise maxima) the rescaled componentwise maximum  $a_n^{-1}\max(X_1,\ldots,X_n)$  converge in distribution as $n\to\infty$ to a jointly $\alpha$-Fr\'echet random vector with exponent function $V(x)=\Lambda([0,x]^c)$, i.e.
\[
\lim_{n\to\infty}\mathbb{P}\left[a_n^{-1}\max(X_1,\ldots,X_n)\leq  x \right] =\exp(-V(x)),\quad x\in [0,\infty)^d\setminus\{0\};
\]
\item[iii)] (excess above threshold) the rescaled exceedance $u^{-1}X$  given that  some component of $X$ exceeds $u>0$  converges in distribution as $u\rightarrow\infty$ to an $\alpha$-Pareto random vector, i.e.
\[
\lim_{u\to\infty}\mathbb{P}\left[ u^{-1}X \not\leq x\mid X\not\leq u1_d\right]=\frac{V(x\vee 1_d)}{V(1_d)},\quad x\in [0,\infty)^d\setminus [0,1]^d; 
\]
\item[iv)] (sample point process) the sample point process $\{a_n^{-1}X_i,\ 1\leq i\leq n\}$ converges in distribution to a Poisson point process on $[0,\infty)^d\setminus\{0\}$ with intensity $\Lambda$.
\end{itemize}

\subsection{A multivariate version of Breiman Lemma}
Before considering its multivariate extension, let us recall the celebrated Breiman Lemma (see Breiman \cite[Proposition 3]{breiman:1965}). 

\begin{lemm}[Breiman lemma]\label{lem:Breiman-univariate}
Let $R$ and $Z$ be independent non negative random variables satisfying either of the following conditions:
\begin{enumerate}
\item[i)] the tail function $1-F$ of $R$ is regularly varying at infinity with index $-\alpha<0$ and $\mathbb{E}[Z^{\alpha+\varepsilon}]<\infty$ for some $\varepsilon>0$;
\item[ii)] $1-F(x)\sim Cx^{-\alpha}$ as $x\to\infty$ for some $C>0$ and $\mathbb{E}[Z^{\alpha}]<\infty$.
\end{enumerate}
Then, the product $RZ$ is regularly varying with index $\alpha$ and
\[
  \mathbb{P}(RZ>x) \sim \mathbb{E}[Z^\alpha]\mathbb{P}(R>x)\quad\mbox{as } x\to\infty.
\]
\end{lemm}

The following multivariate extension of Breiman Lemma follows the line of Davis and Mikosch \cite[section 4.1]{DM08} .
\begin{prop} \label{prop:1}
Let $R$ be a non negative random variable and $Z$ an independent $d$-dimensional random vector. Assume either of the following conditions is satisfied:
\begin{enumerate}
\item[i)] the tail function $1-F$ of $R$ is regularly varying at infinity with index $-\alpha<0$ and $\mathbb{E}[\|Z\|^{\alpha+\varepsilon}]<\infty$ for some $\varepsilon>0$;
\item[ii)] $1-F(x)\sim Cx^{-\alpha}$ as $x\to\infty$ for some $C>0$ and $\mathbb{E}[\|Z\|^{\alpha}]<\infty$.
\end{enumerate}
Then the product $X=RZ$ defines a regularly varying random vector on $[-\infty,\infty]^d\setminus\{0\}$ with index $\alpha$. More precisely, : 
\[
n\mathbb{P}(a_n^{-1}X \in \cdot)\stackrel{M_0}\longrightarrow \Lambda(\cdot),\quad \mbox{ in $M_0(\mathbb{R}^d)$ as $n\to\infty$},
\]
where $a_n$ is the quantile of order $1-1/n$ of $R$ and the limit measure $\Lambda$  is homogeneous of order $\alpha$ and  given by
\begin{equation}\label{eq:Lambda-integral-form}
\Lambda(A)=\int_0^\infty \mathbb{P}(uZ \in A)\alpha u^{-\alpha -1}\mathrm{d}u,\quad A\subset\mathbb{R}^d\setminus\{0\}\ \mbox{Borel}.
\end{equation}
Moreover, in the case when $Z$ is non-negative, $\Lambda$ is supported by $[0,\infty)^d\setminus\{0\}$ and we have
\[
V(x):=\Lambda([0,x]^c)=\mathbb{E}\left[\bigvee_{i=1}^d \left( \frac{Z_i}{x_i} \right)^\alpha\right],\quad x\in [0,+\infty)\setminus\{0\}.
\]
\end{prop}

\begin{ex}\upshape
For example, this applies directly to the multivariate Student distribution with $\nu$ degrees of freedom that is the product of an inverse $\chi^2(\nu)$ distribution (with heavy tail of order $\nu/2$) and an independent multivariate Gaussian distributions (with moments of all orders).  See Nikololoupolos et al. \cite{NJL09} and Section \ref{subsec:examples} below.
\end{ex}

\begin{proof}[Proof of Proposition~\ref{prop:1}]
Consider an arbitrary norm $\|\cdot\|$ on $\mathbb{R}^{d-1}$ and denote by $\mathcal{S}^{d-1}$ the unit sphere. For $x>0$ and $B\subset\mathcal{S}^{d-1}$ Borel, define
\begin{equation}\label{eq:A}
A=\left\{z\in\mathbb{R}^d\ :\ \|z\|>x, \ z/\|z\| \in B \right\}.
\end{equation}
We have, as $n\to\infty$, 
\begin{align}
    n\mathbb{P}(a_n^{-1}X \in A) &= n\mathbb{P}\left( R\|Z\|> a_n x,\ Z/\|Z\|\in B \right)=n\mathbb{P}\left( R\|Z\| \mathrm{1}_{\left\{ Z/\|Z\|  \in B \right\}} >a_n x \right)\nonumber\\
    &\sim n \mathbb{E}\left(\|Z\|^\alpha \mathrm{1}_{\left\{ {Z}/{\|Z\|} \in B\right\}} \right) \mathbb{P}(R>a_n x) \sim \mathbb{E}\left(\|Z\|^\alpha \mathrm{1}_{\left\{ {Z}/{\|Z\|} \in B\right\}} \right) x^{-\alpha} n \mathbb{P}(R>a_n)\nonumber\\
    &\to x^{-\alpha} \mathbb{E}\left(\|Z\|^\alpha \mathrm{1}_{\left\{ {Z}/{\|Z\|} \in B\right\}} \right).\label{eq:v-convergence}
\end{align}
We have used here the univariate Breiman Lemma~\ref{lem:Breiman-univariate} to go from the first to the second line and then the fact that $R$ has a regularly varying tail with index $\alpha>0$ and that $n\mathbb{P}(R>a_n)\to 1$. Using the fact that the sets of the form \eqref{eq:A} form a convergence determining class (Hult and Linskog \cite{hult:lindskog:2006}), we deduce from Equation~\eqref{eq:v-convergence} the $M_0$-convergence $n\mathbb{P}(X/a_n\in\cdot)\stackrel{M_0}\longrightarrow \Lambda(\cdot)$,
where the limit measure $\Lambda$ is characterized by
\begin{equation}\label{eq:Lambda-product-form}
\Lambda(A)=x^{-\alpha} \mathbb{E}\left(\|Z\|^\alpha \mathrm{1}_{\left\{ {Z}/{\|Z\|} \in B\right\}} \right)
\end{equation}
for all set $A$ of the form~\eqref{eq:A}. We then check that $\Lambda$ admits the integral representation~\eqref{eq:Lambda-integral-form}. Computing the right hand side of \eqref{eq:Lambda-integral-form} with $A$ given by \eqref{eq:A}, we get
\begin{align*} 
\int_0^\infty \mathbb{P}(uZ \in A)\alpha u^{-\alpha -1}\mathrm{d}u
    &=\int_0^\infty \mathbb{E}\left(1_{\{u\|Z\|>x, Z/\|Z\| \in B \}} \right) \alpha u^{-\alpha -1}\mathrm{d}u\\
    &=\mathbb{E}\left(1_{ \left\{ {Z}/{\|Z\|}\in B\right\}}\int_0^\infty 1_{ \left\{ u>x/\|Z\|\right\}}  \alpha u^{-\alpha -1} \mathrm{d}u\right)\\
    &=x^{-\alpha}\mathbb{E}\left(\|Z\|^\alpha 1_{ \left\{ {Z}/{\|Z\|}\in B\right\}}  \right)=\Lambda(A).
\end{align*}
Since the sets $A$ of the form \eqref{eq:A} form a determining class,  the integral representation \eqref{eq:Lambda-integral-form} holds for all $A\subset\mathbb{R}^d\setminus\{0\}$ Borel.
We can then check directly that $\Lambda$ is homogeneous: for $v>0$, 
\begin{align*}
\Lambda (vA)&=\int_0^\infty \mathbb{P}(uZ \in vA) \alpha u^{-\alpha-1}\mathrm{d}u
=\int_0^\infty \mathbb{P}(v^{-1}uZ \in A) \alpha u^{-\alpha-1}\mathrm{d}u\\
&=v^{-\alpha} \int_0^\infty \mathbb{P}(uZ \in A) \alpha u^{-\alpha-1}\mathrm{d}u=v^{-\alpha} \Lambda(A),
\end{align*}
where we used the change of variable $u'=u/v$ on the second line.

Finally, when $Z$ is supported by $[0,\infty)^d$,  Equation~\eqref {eq:Lambda-integral-form} implies that $\Lambda$ is supported by $[0,\infty)^d\setminus\{0\}$ and the   tail function $V$ is computed as follows:
\begin{align*}
V(x)&:=\Lambda([0,x]^c)
= \int_{[0,\infty)^d} \mathbb{P}(uZ\notin [0,x])\alpha u^{-\alpha-1}\mathrm{d}u\\
&=\int_{[0,\infty)^d} \mathbb{P}\left(u> Z_ix_i \mbox{ for some $1\leq i\leq d$}\right)\alpha u^{-\alpha-1}\mathrm{d}u\\
&=\int_{[0,\infty)^d} \mathbb{P}\left(u> \min_{1\leq i\leq d}\frac{x_i}{Z_i}\right)\alpha u^{-\alpha-1}\mathrm{d}u
= \mathbb{E}\left[ \left( \min_{1\leq i\leq d}\frac{x_i}{Z_i}\right)^{-\alpha}\right]\\
&=\mathbb{E}\left[\bigvee_{i=1}^d \left( \frac{Z_i}{x_i} \right)^\alpha\right].
\end{align*}
\end{proof}

\begin{prop}\label{prop:2}
If $Z$ has a density $f_Z$, then $\Lambda$ is absolutely continuous with respect to the Lebesgue measure and its Radon-Nikodym derivative is given by
\begin{equation}\label{eq:lambda}
    \lambda(z)=\int_0^\infty f_Z\left(z/u\right) \alpha u^{-d-\alpha-1}\mathrm{d}u.
\end{equation}
and is homogeneous of order $-d-\alpha$, that is
\begin{equation}\label{eq:homogeneous-density}
\lambda(vz)=v^{-d-\alpha}\lambda(z),\quad v>0,\ z\in\mathbb{R}^{d}\setminus\{0\}.
\end{equation}
\end{prop}

\begin{proof}
If $Z$ has a density $f_Z$, the measure $\Lambda$ writes
\begin{align*}
\Lambda(A)&=\int_0^\infty \mathbb{P}(uZ\in A) \alpha u^{-\alpha -1}\mathrm{d}u 
 =\int_0^\infty \int_{\mathbb{R}^d} \mathrm{1}_{\left\{ uz \in A\right\}}f_Z(z) \mathrm{d}z u^{-\alpha-1}\mathrm{d}u\\
&=\int_0^\infty \int_A f_Z\left({z}/{u}\right)\alpha u^{-\alpha -d -1} \mathrm{d}z \mathrm{d}u= \int_A \lambda(z) \mathrm{d}z,
\end{align*}
where we use the change of variable $z'=uz$ and Fubini Theorem. Furthermore, with the change of variable $u'=u/v$,
\[
\lambda (vz)= \int_0^\infty f_Z(vz/u) \alpha u^{-d-\alpha-1}\mathrm{d}u
= v^{-d-\alpha} \int_0^\infty f_Z(z/u) \alpha u^{-d-\alpha-1}\mathrm{d}u=v^{-d-\alpha}\lambda(z).
\]
\end{proof}

\subsection{A copula point of view}
When focusing on the dependence structure, Proposition~\ref{prop:1} can be rephrased in terms of copulas (we refer to Joe \cite{J15} for a background on copulas and Gudendorf and Segers for extreme value copulas \cite{GS10}). Following Krupskii et al. \cite{KJLG17}, we consider here the simple common factor model \begin{equation}\label{eq:comon-factor-model}
X=\alpha E 1_d+ Y
\end{equation}
with $\alpha>0$, $E$  exponentially distributed and, independently, $Y$  a  $d$-dimensional random vector such that $\mathbb{E}[e^{\alpha Y_i}]<\infty$, $i=1,\ldots,d$. The different component of $X$ share the common factor $E$ that  introduces dependence in the extremes, because the components of $Y$ are lighter tailed. Since the exponential distribution has a density, all the components $X_i=\alpha E+Y_i$ have a continuous distribution. Sklar Theorem entails that the copula $C_X$ pertaining to $X$ is uniquely defined by 
\[
C_X(u_1,\ldots,u_d)=F_{X}(F_{X_i}^\leftarrow(u_1),\ldots,F_{X_d}^\leftarrow(u_d)), \quad (u_1,\ldots,u_d)\in[0,1]^d,
\]
where $F_X$ denotes the multivariate cumulative distribution of $X$ and $F_{X_i}^\leftarrow$ the quantile function of component $X_i$.

\begin{prop}\label{prop:copula}
Consider the copula $C_X$ associated to the random vector $X$ defined by \eqref{eq:comon-factor-model}.
Then 
\[
C_X^n(u_1^{1/n},\ldots,u_d^{1/n}) \to C_V(u_1,\ldots,u_d),\quad (u_1,\ldots,u_d)\in[0,1]^d,
\]
where
\[
C_V(u_1,\ldots,u_d)=\exp\left(-V(\sigma_1(-\log u_1)^{1/\alpha},\ldots,\sigma_d(-\log u_d)^{1/\alpha}) \right)
\]
with 
\[
\sigma_i^\alpha=\mathbb{E}[e^{\alpha Y_i}]\quad\mbox{and}\quad  V(x)=\mathbb{E}[\vee_{i=1}^d \frac{e^{\alpha Y_i}}{x_i^\alpha}].
\]
In words, $C_X$ belongs to the domain of attraction of the extreme value copula $C_V$.
\end{prop}
Here, we use the fact that $\exp(\alpha E)$ has an $\alpha$-Pareto distribution but, in view of the proof and the multivariate Breiman Lemma, the result holds  as soon as $\exp(\alpha E)$ has an heavy tail with index $\alpha$ and $(\alpha+\varepsilon)Y_i$ has a finite exponential moment for $i=1,\ldots,d$.

\begin{proof}[Proof of Proposition \ref{prop:copula}]
By Proposition~\ref{prop:1}, $e^X=e^{\alpha E}e^Y$ is regularly varying with exponent function $V$ and hence, the normalized maximum of $n$ independent copies of $X$ converge to an $\alpha$-Fr\'echet vector with distribution function $e^{-V(x)}$. On the other hand, since the exponential transformation operates separately on each component, $e^X$ has copula $C_X$ and  the normalized  maximum of $n$ i.i.d. copies has copula $C_X^n(u_1^{1/n},\ldots,u_d^{1/n})$. It remains to note that $C_V$ is the copula associated with the limiting 
$\alpha$-Fr\'echet vector, where the $i$-th margin as  shape parameter $\alpha$ and scale parameter $\sigma$.
The fact that convergence of pointwise maxima implies convergence of the copula is justified in Deheuvels \cite{D78}.
\end{proof}

\subsection{Examples}\label{subsec:examples}
In this section, we apply Proposition~\ref{prop:1} and consider various models for the various random vector $Z$. For these models, we provide an explicit expression for the limit measure $\Lambda$ that characterizes the regular variation of the product $X=RZ$. Our computations rely on the general form of the density $\lambda$ expressed in Proposition~\ref{prop:2} and technical computations.

\subsubsection{Gaussian case}
The following result states a regular variation result in connection with the extremal-t model, see Nikoloulopoulas et al. \cite{NJL09}.
\begin{prop}\label{prop:gaussian}
In the framework of the multivariate Breiman's lemma, if $Z \sim \mathcal{N}( 0, \Sigma)$, then the limit measure $\Lambda$ has density
\begin{align*}
    \lambda(z)&=\frac{\alpha}{(2\pi)^{d/2} |\Sigma|^{1/2}}\Gamma\left( \frac{\alpha+d}{2}\right) \left(\frac{z^t \Sigma^{-1}z}{2}\right)^{-(\alpha+d)/2},\quad z\in\mathbb{R}^d\setminus\{0\}.
\end{align*}
\end{prop}

\begin{proof} Starting from Eq.~\eqref{eq:lambda} and introducing the Gaussian density, we get
\begin{align*}
    \lambda(z)&= \int_0^\infty f_Z\left(\frac{z}{u}\right)\alpha u^{-\alpha-d-1}\mathrm{d}u\\
    &=\int_0^\infty \frac{1}{\sqrt{2\pi}^d |\Sigma|^{1/2}} \mathrm{exp}\left\{-\frac{1}{2u^2} z^t\Sigma^{-1}z\right\} \alpha u^{-\alpha-d-1}\mathrm{d}u.
\end{align*}
The change of variable $v=1/u$ in the integral yields
\begin{align*}
    \lambda(z)&=(2\pi)^{-d/2} |\Sigma|^{-1/2} \alpha \int_0^\infty \mathrm{exp}\left\{-\frac{v^2}{2} z^t\Sigma^{-1}z \right\} u^{\alpha+d-1} \mathrm{d}u\\
    &=(2\pi)^{-d/2} |\Sigma|^{-1/2} \alpha \frac{2}{z^t \Sigma^{-1}z} \int_0^\infty \frac{z^t\Sigma z}{2}\mathrm{exp}\left\{-\frac{v^2}{2} z^t\Sigma^{-1}z \right\} u^{\alpha+d-1} \mathrm{d}u\\
    &=(2\pi)^{-d/2} |\Sigma|^{-1/2} \alpha \frac{2}{z^t \Sigma^{-1}z} \mathbb{E}\left[ X^{\alpha+d-2}\right]
\end{align*}
where $X$ has a Weibull distribution with shape parameter equal to $2$ and scale parameter equal to $\sqrt{2/(z^t \Sigma^{-1}z)}$. We deduce 
\[
    \mathbb{E}\left[ X^{\alpha+d-2}\right]=\left(\frac{z^t \Sigma^{-1}z}{2}\right)^{-(\alpha+d-2)/2}\Gamma\left( \frac{\alpha+d}{2}\right)
\]
and we obtain the claimed formula for $\lambda(z)$. 
\end{proof}

\subsubsection{Log-normal case}
The case of log-normal spectral functions is connected with the H\"usler-Resii model \cite{HR89}, see also Wadsworth and Tawn \cite{WT14}.
\begin{prop}\label{prop:lognormal}
In the framework of the multivariate Breiman's lemma, if $Z \sim \mathcal{LN}(m,\Sigma)$ with $\Sigma$ definite positive, then the limit measure $\Lambda$ has density
\begin{align*}
   \lambda(z)= C \,\exp\left\{-\frac{1}{2} \mathrm{log}z^t Q \mathrm{log}z + L \mathrm{log}z \right\} \prod\limits_{i=1}^d z_i^{-1},\quad z\in(0,\infty)^d,  
\end{align*}
where 
\begin{align}
    &C= \frac{\alpha}{(2\pi)^{(d-1)/2}|\Sigma|^{-1/2}\sqrt{\mathrm{1}_d^t \Sigma^{-1}\mathrm{1}_d}} \exp\left\{-\frac{1}{2} m^t \Sigma^{-1} m +\frac{1}{2} \frac{(m^t\Sigma^{-1}\mathrm{1}_d+\alpha)^2}{\mathrm{1}_d^t \Sigma^{-1}\mathrm{1}_d}\right\},\nonumber\\
    &Q=\Sigma^{-1} -\frac{\Sigma^{-1}\mathrm{1}_d \mathrm{1}_d^t  \Sigma^{-1}}{\mathrm{1}_d^t \Sigma^{-1}\mathrm{1}_d}, \label{2.5.1}\\
    &l=\left( m^t -\frac{\alpha+m^t\Sigma^{-1}\mathrm{1}_d}{\mathrm{1}_d^t \Sigma^{-1}\mathrm{1}_d}\mathrm{1}_d^t\right) \Sigma^{-1}.\label{2.5.2}
\end{align}
and 
\begin{align*}
    V(x)=\frac{C (2\pi)^{(d-1)/2}}{\alpha}\sum_{i=1}^d x_i^{-\alpha} |Q_{-i}|^{-1/2}\exp \left\{\frac{1}{2}l^T_{-i}Q_{-i}^{-1}l_{-i} \right\} \Phi_{d-1}\left(\log \frac{x_{-i}}{x_i}; Q_{-i}^{-1}l_{-i}, Q_{-i}^{-1} \right).
\end{align*}
\end{prop}

\begin{proof} Starting from Eq.~\eqref{eq:lambda} and introducing the log-normal Gaussian density, we get
\begin{align*}
    \lambda(z)&=\int_0^\infty f_{Z} \left( \frac{z}{u} \right) \alpha u^{-\alpha-d-1} \mathrm{d}u\\
    &=\int_0^\infty \prod\limits_{i=1}^d z_i^{-1} \alpha |\Sigma|^{-1/2} (2\pi)^{-d/2} \mathrm{exp}\left\{-\frac{1}{2}(\mathrm{log}(z) - \mathrm{log}(u) \mathrm{1}_d-m)^t\Sigma^{-1}(\mathrm{log}(z) - \mathrm{log}(u) \mathrm{1}_d-m) \right\} \\& \ \ u^{-\alpha -1}\mathrm{d}u
\end{align*}
The change of variable $    v=\mathrm{log} (u)$
yields
\begin{align*}
    \lambda(z)&=\alpha |\Sigma|^{-1/2}(2\pi)^{-d/2}\prod\limits_{i=1}^d z_i^{-1} \int_{-\infty}^{\infty} \exp\left\{\mathrm{P}(v) \right\} \mathrm{d}v
\end{align*}
with 
\begin{align*}
    \mathrm{P}(v)&= -\frac{1}{2}(\mathrm{log}(z)- v \mathrm{1}_d -m)^t\Sigma^{-1}(\mathrm{log}(z)- v \mathrm{1}_d -m)-\alpha v\\
    &= -\frac{1}{2} \mathrm{1}_d^t\Sigma^{-1}\mathrm{1}_d v^2 + \left(\mathrm{log}z^t\Sigma^{-1}\mathrm{1}_d- m^t\Sigma^{-1}\mathrm{1}_d -\alpha \right) v -\frac{1}{2}\mathrm{log}z^t \Sigma^{-1}\mathrm{log}z -\frac{1}{2}m^t\Sigma^{-1}m  + \mathrm{log}z^t \Sigma^{-1} m\\
    &=-\frac{1}{2}C_1 v^2+ C_2v+ C_3.
\end{align*}
Recognizing a Gaussian integral, we get with $X\sim\mathcal{N}(0,C_1^{-1})$,
\[
    \int_{-\infty}^\infty \mathrm{exp}\left\{ P(v)\right\} \mathrm{d}v =\sqrt{\frac{2\pi}{C_1}}e^{C_3}\mathbb{E}[\mathrm{exp}\{C_2 X\}]
    =\sqrt{\frac{2\pi}{C_1}}e^{C_3}e^{\frac{C_2^2}{2C_1}}.
\]
We deduce the claimed formula for $\lambda(z)$ after some straightforward simplifications.
\end{proof}

\subsubsection{Independent Fr\'echet case}
The case of independent spectral components is related to the logistic model  \cite{GS10}.
\begin{prop}[Frechet case]
Suppose $Z=(Z_1,\dots, Z_d)$ with $Z_i \sim \mathrm{Frechet}(\lambda_i,\beta)$ independent with $\beta>\alpha$. Then, the limit measure $\Lambda$ in multivariate Breiman's lemma has density 
\begin{align*}
    \lambda(z)=\alpha \beta^{d-1}\Gamma(d-\alpha/\beta)  \prod_{i=1}^d \frac{z_i^{-\beta-1}}{\lambda_i^{-\beta}}\left(\sum\limits_{i=1}^d \left(\frac{z_i}{\lambda_i}\right)^{-\beta} \right)^{(\alpha+1)/\beta-d} 
\end{align*}
with $\Gamma$ the Gamma function and 
\[
V(x):=\Lambda([0,x]^c)=\Gamma\left(1-\frac{\alpha}{\beta} \right)\left(\sum\limits_{i=1}^d\left(\frac{x_i}{\lambda_i}\right)^{-\beta}\right)^{\alpha/\beta}.
\]
\end{prop}
\begin{proof} Starting from Eq.~\eqref{eq:lambda} and introducing the product Fr\'echet density yields
\begin{align*}
    \lambda(z)&=\int_0^\infty f_{Z} \left( \frac{z}{u} \right) \alpha u^{-\alpha-d-1} \mathrm{d}u\\
    &=\int_0^\infty  \prod_{i=1}^d \left(z_i^{-\beta-1} u^{\beta+1} \beta \lambda_i^{\beta} \mathrm{exp}\left\{ -(z_i/\lambda_iu)^{-\beta} \right\}\right) \alpha u^{-\alpha-d-1} \mathrm{d}u \\
    &=\alpha \beta^d \prod_{i=1}^d \frac{z_i^{-\beta-1}}{\lambda_i^{-\beta}} \int_0^\infty u^{-\alpha+\beta d-1}\mathrm{exp}\left\{ -u^\beta \sum\limits_{i=1}^d \frac{z_i}{\lambda_i}^{-\beta}\right\}\mathrm{d}u.
\end{align*}
The change of variable $v=u^{\beta} \sum\limits_{i=1}^d\left( \frac{z_i}{\lambda_i}\right)^{-\beta}$ in the integral gives
\begin{align*}
    \lambda(z)=\alpha \beta^d \prod_{i=1}^d \frac{z_i^{-\beta-1}}{\lambda_i^{-\beta}}\frac{1}{\beta}\left(\sum\limits_{i=1}^d \left(\frac{z_i}{\lambda_i}\right)^{-\beta} \right)^{(\alpha+1)/\beta-d} \int_0^\infty e^{-v} v^{d -\alpha/\beta -1} \mathrm{d}v
\end{align*}
The last integral is the definition of the Gamma function $\Gamma(d-\alpha/\beta)$.
Proposition \ref{prop:1} gives
\begin{align*}
    V(x)&=\mathbb{E}\left[\bigvee_{i=1}^d \left(\frac{Z_i}{x_i}\right)^{\alpha}\right]\\
    &=\int_{0}^{\infty} \mathbb{P}\left( \bigvee_{i=1}^d \left(\frac{Z_i}{x_i}\right)^{\alpha}>x\right) \mathrm{d}x\\
    &=\int_0^{\infty} 1-\prod_{i=1}^d \mathbb{P}\left(\left(\frac{Z_i   }{x_i}\right)^{\alpha}\le x\right) \mathrm{d}x.
\end{align*}
Introducing the Fréchet density function yields
\begin{align*}
    V(x)=\int_{0}^{\infty} 1-\exp\left(-x^{-{\beta/\alpha}}\sum\limits_{i=1}^d\left(\frac{x_i}{\lambda_i} \right)^{-\beta} \right)\mathrm{d}x.
\end{align*}
The change of variable $y=x\left(\sum\limits_{i=1}^d \left(\frac{x_i}{\lambda_i}\right)^{-\beta}\right)^{-\alpha/\beta}$ gives
\begin{align*}
    V(x)=\left(\sum\limits_{i=1}^d\left(\frac{x_i}{\lambda_i}\right)^{-\beta}\right)^{\alpha/\beta}\int_{0}^{\infty} 1-\exp\left(-y^{-\beta/\alpha} \right)\mathrm{d}y
\end{align*}
The last integral correspond to the expectation of a Fréchet$(1,\beta/\alpha)$ and therefore, assuming $\beta>\alpha$,  we have the result.
\end{proof}

\subsubsection{Independent Weibull case}
The case of independent spectral components is related to the negative logistic model  \cite{GS10}.
\begin{prop}[Weibull case]
Suppose $Z=(Z_1,\dots, Z_d)$ with $Z_i \sim \mathrm{Weibull}(\lambda_i, \beta)$ independent with $\alpha>\beta$. Then the limit measure $\Lambda$ in multivariate Breiman's Lemma has density
\begin{align*}
    \lambda(z)=\alpha \beta^{d-1} \Gamma(d+\alpha/\beta)\left(\sum\limits_{i=1}^d \left( \frac{z_i}{\lambda_i} \right)^\beta \right)^{-(\alpha+1)/\beta-d} \prod_{i=1}^d  \frac{z_i^{\beta-1}}{\lambda_i^{\beta}}
\end{align*}
and 
\[
V(x):=\Lambda([0,x]^c)=\Gamma\left(1+\frac{\alpha}{\beta} \right)\sum_{\emptyset \ne J \subset \{ 1,\cdots,d\}} (-1)^{|J|+1}\left(\sum\limits_{j\in J}\left(\frac{x_j}{\lambda_j}\right)^\beta\right)^{-\alpha/\beta}.
\]
\end{prop}
\begin{proof} Starting from Eq.~\eqref{eq:lambda} and introducing the product Weibull density yields
\begin{align*}
    \lambda(z)&=\int_0^\infty f_{Z} \left( \frac{z}{u} \right) \alpha u^{-\alpha-d-1} \mathrm{d}u\\
    &=\int_0^\infty \prod_{i=1}^d \left( \frac{\beta}{\lambda_i}\left(\frac{z_i}{u\lambda_i} \right)^{\beta-1} \mathrm{exp}\left\{-\left(\frac{z_i}{u \lambda_i} \right)^\beta \right\}\right) \alpha u^{-\alpha-d-1}\mathrm{d}u\\
    &=\alpha \beta^d\prod_{i=1}^d \left(\frac{1}{\lambda_i}\left(\frac{z_i}{\lambda_i}\right)^{\beta-1} \right) \int_0^\infty \mathrm{exp}\left\{-u^{-\beta}\left(\sum\limits_{i=1}^d\left(\frac{z_i}{\lambda_i}\right)^{\beta} \right) \right\} u^{-\alpha-\beta d-1} \mathrm{d}u.
\end{align*}
The change of variable $v=u^{-\beta}\left( \sum\limits_{i=1}^d \left( \frac{z_i}{\lambda_i}\right)^\beta \right)$ in the integral gives
\begin{align*}
    \lambda(z)=\alpha \beta^{d-1} \left(\sum\limits_{i=1}^d \left( \frac{z_i}{\lambda_i} \right)^\beta \right)^{-(\alpha+1)/\beta-d} \prod_{i=1}^d  \frac{z_i^{\beta-1}}{\lambda_i^{\beta}} \int_0^\infty e^{-v}v^{\frac{\alpha}{\beta}+d-1}\mathrm{d}v.
\end{align*}
Proposition \ref{prop:1} yields
\begin{align*}
V(x)&=\mathbb{E}\left[ \bigvee_{i=1}^d \left(\frac{Z_i}{x_i}\right)^{\alpha}\right]\\
&=\int_{0}^{\infty} 1-\prod_{i=1}^d \mathbb{P}\left( \left(\frac{Z_i}{x_i}\right)^{\alpha}\le x\right) \mathrm{d}x.
\end{align*}
Introducing the Weibull density function yields
\begin{align*}
V(x)&=\int_{0}^\infty 1-\prod_{i=1}^d \left( 1 -\exp\left(-x^{\beta/\alpha}\left(\frac{x_i}{\lambda_i}\right)^\beta \right)\right)\mathrm{d}x\\
&=\int_{0}^\infty \sum_{\emptyset \ne J\subset \{ 1,\cdots,d\}} (-1)^{|J|+1}\exp\left(-x^{\beta/\alpha}\sum\limits_{j \in J}\left(\frac{x_j}{\lambda_j}\right)^\beta \right)\mathrm{d}x.
\end{align*}
The change of variable $y=x\left(\sum\limits_{j\in J}\left(\frac{x_j}{\lambda_j}\right)^\beta\right)^{\alpha/\beta}$ yields
\begin{align*}
    V(x)=\sum_{\emptyset \ne J \subset \{ 1,\cdots,d\}} (-1)^{|J|+1}\left(\sum\limits_{j\in J}\left(\frac{x_j}{\lambda_j}\right)^\beta\right)^{-\alpha/\beta}\int_0^\infty  \exp\left(-y^{\beta/\alpha}\right)\mathrm{d}y.
\end{align*}
The last integral correspond to the expectation of a Weibull$(1,\beta/\alpha)$.
\end{proof}

\subsubsection{Independent Gamma case}
This last example is related to the max-stable model with Dirichlet spectral density. $\beta_i\equiv 1$, the restriction of $\lambda$ on the simplex is proportional to the Dirichlet density.
\begin{prop}[Gamma case]
Suppose $Z=(Z_1,\dots, Z_d)$ with $Z_i \sim \Gamma(\theta_i, \beta_i)$ independent. Then
\begin{align*}
    \lambda(z)&=\alpha \Gamma\left(\alpha+ \sum_{i=1}^d \theta_i\right)\left(\sum\limits_{i=1}^d \beta_i z_i \right)^{-\sum_{i=1}^d \theta_i -\alpha} \prod_{i=1}^d  \left(\frac{\beta_i^{\theta_i} z_i^{\theta_i-1}}{\Gamma(\theta_i)} \right)
\end{align*}
\end{prop}
\begin{proof}
    \begin{align*}
        \lambda(z)&=\int_0^\infty \prod_{i=1}^d\left(\frac{\beta_i^{\theta_i}}{\Gamma(\theta_i)} \left(\frac{z_i}{u}\right)^{\theta_i-1} e^{-\beta_i z_i/u}   \right) \alpha u^{-\alpha-1-d}\mathrm{d}u\\
        &=\alpha \prod_{i=1}^d  \left(\frac{\beta_i^{\theta_i} z_i^{\theta_i-1}}{\Gamma(\theta_i)} \right) \int_0^\infty u^{-\sum_{i=1}^d \theta_i -\alpha -1} \mathrm{exp}\left\{-u^{-1}\sum_{i=1}^d \beta_i z_i \right\} \mathrm{d}u.
    \end{align*}
    Setting $v= u^{-1}\sum_{i=1}^d \beta_i z_i$, we obtain
    \begin{align*}
        \lambda(z)&= \alpha \left(\sum\limits_{i=1}^d \beta_i z_i \right)^{-\sum_{i=1}^d \theta_i -\alpha}  \prod_{i=1}^d  \left(\frac{\beta_i^{\theta_i} z_i^{\theta_i-1}}{\Gamma(\theta_i)} \right) \int_0^\infty e^{-v} v^{\sum_{i=1}^d \theta_i +\alpha-1} \mathrm{d}u.
    \end{align*}
\end{proof}

\subsection{Non standard regular variations}
 Following Resnick \cite{R07}, non-standard multivariate regular variations correspond to different tail index for the different components. Proposition \ref{prop:1} has a simple extension to this case.
        \begin{prop}\label{prop-1-gen}
        Let $R$ be a non negative heavy-tailed random variable with index $1$, $\alpha=(\alpha_1, \cdots,\alpha_d)\in (0,\infty)^d$ and  $Z=(Z_1,\cdots,Z_d)$ a d-dimensional random vector such that $\mathbb{E}\vert Z_i\vert^{\alpha_i + \varepsilon}<\infty$ for some $\varepsilon>0$. Then the product $X=(R^{1/\alpha_1}Z_1,\cdots,R^{1/\alpha_d}Z_d)=R^{1/\alpha }Z$ satisfies
        \[
        n\mathbb{P}(a_n^{-1/\alpha}X\in \cdot)\stackrel{M_0}\longrightarrow \Lambda(\cdot)
        \]
        where $a_n$ is the quantile of order $1-n^{-1}$ of $R$ and the limit measure $\Lambda$ satisfies
        \begin{align}
            \Lambda(A)=\int_0^\infty \mathbb{P}\left( u ^{-1/\alpha}Z \in A \right) \mathrm{d}u, \ \ A \subset \mathbb{R}^d\setminus\{0\}\ measurable.
        \end{align}
        \end{prop}
        \begin{proof}
        Proposition \ref{prop:1} for $\tilde X=RZ^\alpha=(RZ_1^{\alpha_1},\ldots,RZ_d^{\alpha_d})$ yields the regular variations for $\tilde X$. Then, the change of variable  $X=\tilde X^{1/\alpha}$ together with the continuous mapping theorem for $M_0$-convergence \cite{hult:lindskog:2006} imply the non-standard regular variations stated in Proposition~\ref{prop-1-gen} . 
        \end{proof}

\section{The H\"usler-Reiss Pareto model}

\subsection{Definition and transformation properties}

Motivated by Proposition~\ref{prop:lognormal}, we introduce the family of H\"usler-Reiss Pareto distributions and study their properties. The main reason why we focus on that particular class is that it enjoys an exponential family property, see Bandorff-Nielsen \cite{BN17}.

\begin{defi}\label{def:lognormalPareto}
Let $d\geq 2$, $a=(a_1,\ldots,a_d)\in(0,\infty)^d$,  $Q\in \mathbb{R}^{d\times d}$ a  symmetric positive semi-definite matrix such that $\mathrm{Ker}\,Q=\mathrm{vect}(1_d)$ and $l\in \mathbb{R}^{d}$ such that  $l^T1_d<0$. The H\"usler-Reiss Pareto model on $[0,\infty)^d\setminus[0,a]$ with parameters $(Q,l)$ is defined by the density
\[
f_{a}(z;Q,l)=\frac{1}{C_a(Q,l)}\exp\left(-\frac{1}{2}\log z^T Q \log z+l^T\log z \right)\left(\prod_{i=1}^d z_i^{-1}\right)1_{\{z\nleq a\}}, \quad z\in (0,\infty)^d,
\]
with $C_a(Q,l)$ the normalization constant. We call $\alpha=-l^T1_d>0$ the exponent of the Pareto distribution $f_{a}(z;Q,l)$.\\
We write $Z\rightsquigarrow\mathrm{HRPar}_{a}(Q,l)$ for a  random vector $Z$ with density $f_{a}(z;Q,l)$.
 \end{defi}

\begin{rema}\label{remarklog}\upshape
The H\"usler-Reiss Pareto model is closely connected with the exponent measure $\lambda$ obtained in Proposition~3.2. Indeed, the parameters $Q$ and $l$ introduced there satisfy the constraint stated in Definition~\ref{def:lognormalPareto}.  The symmetric semi-definite positive matrix $Q$ satisfies 
\[
Q1_{d}=\left(\Sigma^{-1} -\frac{\Sigma^{-1}\mathrm{1}_d \mathrm{1}_d^T \Sigma^{-1}}{\mathrm{1}_d^T \Sigma^{-1}\mathrm{1}_d}\right)1_d=0
\]
and, for all vector $x\in\mathbb{R}^d\setminus\{0\}$ such that $x^T\Sigma^{-1}1_d=0$, we have $x^TQx>0$ whence we deduce $\mathrm{Ker}\,Q=\mathrm{vect}(1_d)$. As for $l$, we check readily
\[
l^T1_d=\left( m^T -\frac{\alpha+m^T\Sigma^{-1}\mathrm{1}_d}{\mathrm{1}_d^T \Sigma^{-1}\mathrm{1}_d}\mathrm{1}_d^T\right) \Sigma^{-1}1_d=-\alpha<0.
\]
Conversely, for all $(Q,l)$ as in Definition~\ref{def:lognormalPareto}, there exist  (non unique) 
$\Sigma\in\mathbb{R}^{d\times d}$ and $m\in\mathbb{R}^d$ such that Equations \eqref{2.5.1} and \eqref{2.5.2} are satisfied.
\end{rema}

\begin{ex}\label{ex:dim-2}\upshape
In dimension $d=2$, the model parameters are
\[
Q=\left(\begin{array}{rr}c&-c\\-c&c\end{array}\right)\quad \mbox{and}\quad l=\left(\begin{array}{c}l_1\\l_2\end{array}\right) \quad \mbox{with } q>0,\ l_1+l_2<0. 
\]
The exponent is $\alpha=-(l_1+l_2)>0$ and 
\[
f_a(z;Q,l)=\frac{1}{C_a(Q,l)}\exp\left(-\frac{c}{2}(\log z_1-\log z_2)^2+l_1\log z_1+l_2\log z_2 \right)\frac{1}{z_1z_2}1_{\{z\nleq a\}}.
\]
\end{ex}

\bigskip Interestingly, Hüsler-Reiss Pareto distributions inherit from log-normal distributions a stability property under scale and power transformations.
\begin{prop}\label{scalepow}
Let $Z\rightsquigarrow\mathrm{HRPar}_{a}(Q,l)$.
\begin{itemize}
\item[(i)] For all $u\in(0,\infty)^d$, $uZ\rightsquigarrow\mathrm{HRPar}_{ua}(Q,l+Q\log u)$.
\item[(ii)] For all $\beta>0$, $Z^{\beta}\rightsquigarrow\mathrm{HRPar}_{a^{\beta}}(\beta^{-2}Q,\beta^{-1}l)$. In particular, if $Z$ has exponent $\alpha$,  $Z^\beta$ has exponent $\alpha/\beta$.
\end{itemize}
\end{prop}
\begin{proof}
The change of variable $\tilde{z}=uz$ implies
\[
\mathbb{P}(uZ \in A)=\int_A f_a(\tilde{z}/u;Q,l)\prod_{i=1}^du_i^{-1}\mathrm{d}\tilde{z}.
\]    
Simple computations show that 
\begin{align*}
&f_a(z/u;Q,l)\prod_{i=1}^du_i^{-1}\\
 &\quad=\frac{1}{C_a(Q,l)}\exp\left(-\frac{1}{2}(\log z-\log u)^T Q (\log z-\log u) +l^T (\log z -\log u)\right)\left(\prod_{i=1}^d z_i^{-1}\right)1_{\{z \nleq ua\}}\\
&\quad = \frac{C_{ua}(Q,l+Q\log u)}{\exp\left(\frac{1}{2} \log u^T Q \log u +l^T \log u\right)C_a(Q,l)}f_{ua}(z;Q,l+Q\log u)
\end{align*}    
This proves (i) as well as the equality
\[
C_{ua}(Q,l+Q\log u)=\exp\left(\frac{1}{2} \log u^T Q \log u +l^T \log u\right)C_a(Q,l).
\]
Using a similar reasoning, the change of variable $\tilde{z}=z^{\beta}$ yields
\[
\mathbb{P}(Z^\beta \in A)=\int_A f_a(\tilde{z}^{1/\beta})\prod_{i=1}^d \frac{1}{\beta} \tilde{z}_i^{1/\beta-1}\mathrm{d}\tilde{z}
\]
and simple computations show that
\begin{align*}
    &f_a(z^{1/\beta})\prod_{i=1}^d \frac{1}{\beta} z_i^{1/\beta-1} \\
    &\quad =\frac{1}{C_a(Q,l) \beta^d}\exp\left(-\frac{1}{2}\log z^T \beta^{-1} Q \beta^{-1} \log z + l^T \beta^{-1} \log z\right)\left(\prod_{i=1}^d z_i^{-1}\right)1_{\{z\nleq a^\beta\}}\\
    &\quad = \frac{C_{a^\beta}(\beta^{-2},\beta^{-1}l)}{C_a(Q,l) \beta^d} f_{a^{\beta}}(z; \beta^{-2}Q, \beta^{-1}l) 
\end{align*}
This implies (ii) as well as the equality
\[
C_{a^\beta}(\beta^{-2}Q,\beta^{-1}l)=\beta^d C_a(Q,l).
\]
\end{proof}

\begin{rem}\upshape
As a consequence of Proposition~\ref{scalepow}, Hüsler-Reiss Pareto vectors with $a=1_d$ and $\alpha=1$ are particularly important, especially for simulation. Indeed, the random vector $Z\rightsquigarrow\mathrm{HRPar}_{a}(Q,l)$ with exponent $\alpha=-l^T1_d$ satisfies $Z\stackrel{d}=a\widetilde{Z}^{1/\alpha}$ where the random vector $\widetilde{Z}\rightsquigarrow\mathrm{HRPar}_{1_d}\left(\alpha^{-2}Q,\alpha^{-1}(l-Q\log a)\right)$ takes values in $[0,\infty)^d\setminus[0,1_d]$ and  has exponent $1$.
\end{rem}

\begin{rem}\upshape
The following equalities on the normalizing constant seen in the proof of Proposition \eqref{scalepow} are worth noting:
\[
C_{ua}(Q,l+Q\log u) = \exp\left(\frac12 \log u^T Q\log u+l^T\log u \right)C_a(Q,l) 
\]
and
\[
C_{a^{\beta}}(\beta^{-2}Q,\beta^{-1}l) = \beta^{d}C_a(Q,l).
\]
As a consequence, we will often assume without loss of generality that $a=1_d$. The general case $a\in(0,\infty)^d$ follows with the relation
\[
C_{a}(Q,l) = \exp\left(-\frac12 \log a^T Q\log a+l^T\log a \right)C_{1_d}(Q,l-Q\log a). 
\]

\end{rem}

\subsection{Exponential family properties}
An important property of the Hüsler-Reiss Pareto distributions introduced above is to form an exponential family. Let $E$ be an euclidean space with dot product $\langle\cdot,\cdot\rangle$. A parametric family of densities $(f(z;\theta))_{\theta\in\Theta}$ with $\Theta\subset E$ is  a \emph{canonical exponential family} if it can be written in the form
\begin{equation}\label{eq:exponential-model}
f(z;\theta)=\frac{1}{C(\theta)}e^{\langle \theta,T(z)\rangle}h(z),\quad z\in\mathbb{R}^d,
\end{equation}
where $T:\quad\mathbb{R}^d\to E$ is the \emph{natural sufficient statistic}. The exponential family is called a \emph{full} exponential family if 
\[
\Theta=\left\{t\in E\ :\ \int_{\mathbb{R}^d}e^{\langle t,T(z)\rangle} h(z)\,\mbox{d$z$} <\infty\right\}
\] 
is not contained in a strict subspace of $E$. For a detailed account on exponential family, the reader should refer to Barndorff-Nielsen \cite{MR3221776}.

\medskip
Our main result in this section is the following Theorem.
\begin{theo}\label{theo:Pareto-exponential-family}
Consider the $d(d+1)/2$-dimensional Euclidean space
\[
E=\{(A,b)\in\mathbb{R}^{d\times d}\times\mathbb{R}^d\ :\ A^T=A,\ A1_d=0\}
\]
with inner product 
\[
\langle(A,a),(A',a')\rangle=\sum_{1\leq i,j\leq d}A_{i,j}A'_{i,j}+\sum_{1\leq k\leq d}a_ka'_k.
\]
Define
\[
\Theta=\left\{(Q,l)\in E\ : Q \mbox{ semi definite positive},\ \mathrm{Ker}\,Q=\mathrm{vect}(1_d),\ l^T1_d<0\right \}.
\]
For all fixed $a\in(0,\infty)^d$, the H\"usler-Reiss Pareto distributions $(f_a(z;\theta))_{\theta\in\Theta}$ form a full canonical exponential family with parameter $\theta=(Q,l)\in\Theta$ and sufficient statistic
\begin{align}\label{T}
    T(z)= \left( -\frac{1}{2}\left(\log z-\overline{\log z}\right)\left(\log z-\overline{\log z}\right)^T,\log z\right),
\end{align}
with $\overline{\log z}=d^{-1}(1_d^T \log z) 1_d$.
\end{theo}

\begin{proof} Without loss of generality, let $a=1_d$. Consider the intensity function
\begin{equation}\label{eq:tildelambda}
\tilde\lambda(z)=\exp\left(-\frac{1}{2}\log z^T Q \log z+l^T\log z \right)\left(\prod_{i=1}^d z_i^{-1}\right),\quad z\in (0,\infty)^d,
\end{equation}
The symmetric matrix $Q$ can be diagonalized in an orthonormal basis $Q=U\Delta U^T$ with $\Delta=\mathrm{diag}(\lambda_1,\ldots,\lambda_d)$ and $U$ orthonormal.
Thanks to the condition $Q1_d=0$, we can suppose $\lambda_1=0$ and the first column of $U$ is equal to $U_1=1_d/\sqrt{d}$. Denote by $\Delta_{-1}$ (resp. $v_{-1}$) the matrix $\Delta$ (resp. vector $v$) with its first row and column removed (resp. first component removed), $\tilde{U}$ the $d\times( d-1)$ matrix obtained by removing the first column of $U$. The change of variable $\log z=Uv$ gives
\begin{align*}
    \int_{(0,\infty)^d} 1_{\{z\nleq a\}}\tilde{\lambda}(z) \mathrm{d}z=\int_{\mathbb{R}^d} \exp\left(-\frac{1}{2}v_{-1}^T \Delta_{-1} v_{-1} +l^T\tilde{U}v_{-1}+l^TU_1v_1\right) 1_{v \in A} \mathrm{d}v
\end{align*}
where $A$ equals
\[
A= \left\{ v \in \mathbb{R}^d : v_1 1_d \nleq -\tilde{U}v_{-1}\right\}=\left\{v\in\mathbb{R}^d:v_1>a(v_{-1}); a(v_{-1})=\min_{i} -\sum\limits_{j=1}^{d-1}\tilde{U}_{ij}v_{j+1}\right\}.
\] 
By Fubini theorem,
\begin{align*}
\int_{(0,\infty)^d} 1_{\{z\nleq a\}}\tilde{\lambda}(z)\mathrm{d}z=\int_{\mathbb{R}^{d-1}}\exp\left(-\frac{1}{2}v_{-1}^T \Delta_{-1}v_{-1} + l^T \tilde{U}v_{-1} \right) \int_{a(v_{-1})}^{\infty} \exp\left(l^T U_1 v_1\right) \mathrm{d}v_1 \mathrm{d}v_{-1}.
\end{align*}
The inner integral with respect to $v_1$ converges if and only if $l^TU_1<0$ and then
\begin{align*}
   \int_{(0,\infty)^d} 1_{\{z\nleq a\}} \tilde{\lambda}(z)\mathrm{d}z=\int_{\mathbb{R}^{d-1}}\exp\left( -\frac{1}{2}v_{-1}^T\Delta_{-1} v_{-1} +l^T\tilde{U}v_{-1}+ l^T U_1 a(v_{-1})\right) \mathrm{d}v_{-1}
\end{align*}
is finite if and only if $\Delta_{-1}$ is positive definite. This proves that the integral converge if and only if $(Q,l)\in\Theta$ and that the exponential family is full.
\end{proof}

In a general exponential model \eqref{eq:exponential-model}, the logarithm of the normalisation constant $C(\theta)$ is related to the cumulant generating function of the natural statistics $T$ by the relation
\[
\log\mathbb E_\theta\left[e^{\langle t,T(z)\rangle} \right]= \log C(\theta+t)-\log C(\theta),\quad \theta,\theta+t\in\Theta.
\]
If $\theta$ is an interior point of $\Theta$, this implies 
\[
\mathbb{E}_\theta[T(z)]=\frac{\partial \log C}{\partial\theta}(\theta) 
\quad\mbox{and}\quad 
\mathrm{Var}_\theta[T(z)]=\frac{\partial^2 \log C}{\partial\theta\partial\theta^T}(\theta).\tag{3.6.1}\label{3.6.1}
\]
The computation of the normalization constant $C(\theta)$ is hence particularly important.
\begin{prop}\label{prop:C_a(Q,l)}
In the H\"usler-Reiss Pareto model described in Theorem~\ref{theo:Pareto-exponential-family}, we have
\[
C_a(Q,l)= (2\pi)^{(d-1)/2}\frac{1}{\alpha}\sum_{i=1}^d a_i^{-\alpha}  \mathrm{det}(Q_{-i})^{-1/2}\exp\left\{ \frac{1}{2}l_{-i}^TQ_{-i}^{-1}l_{-i}\right\} \Phi_{d-1}\left( \log \frac{a_{-i}}{a_i};Q_{-i}^{-1}l_{-i},Q_{-i}^{-1}\right),
\]
where $\alpha=-1_d^Tl$, the notation $l_{-i}$ (resp. $a_{-i}$) denotes the vector $l$ (resp. $a$) with its $i$th component removed, $Q_{-i}$  the matrix $Q$ with its $i$th column and row removed and $\Phi_{d}(z;m,\Sigma)$ denotes the cumulative distributive function at $z$ of a $d$-dimensional multivariate Gaussian distribution with mean $m$ and covariance $\Sigma$.
\end{prop}
The expression for $C_a(Q,l)$ was first established by Huser and Davison \cite{HD13}. We provide here a direct proof that will be needed for further reference (proof of Proposition~\ref{prop:simulation}).
\begin{proof}
  With $\tilde\lambda$ the function defined by Equation~\eqref{eq:tildelambda},
the normalization constant $C_a(Q,l)$ is given by
		\[
		C_a(Q,l)=		\int_{[0,a]^c}\tilde \lambda(z)\,\mathrm{d}z.
		\]
Since $[0,a]^c=\cup_{i=1}^d A_i$ with
\[
A_i=\left\{ z\in \mathbb{R}^d \ :\  z_i>a_i , z_{-i}/z_i\leq  a_{-i}/a_i\right\},\quad  i=1,\ldots,d,
\]
we have
		\[
		C_a(Q,l)=		\sum_{i=1}^d \int_{A_i}\tilde \lambda(z)\,\mathrm{d}z.
		\]
Using the homogeneity relation~\eqref{eq:homogeneous-density} with $x=z_i$, we get
\[
\tilde \lambda(z)=\tilde \lambda(z_i\,z/z_i)=z_i^{-d-\alpha}\tilde\lambda(z/z_i).
\]
Since the $i$th component of $z/z_i$ is equal to $1$, we have also
\[
\tilde \lambda(z/z_i)=\exp\left(-\frac{1}{2}\log \tilde{z}_{-i}^T Q_{-i} \log \tilde{z}_{-i}+l_{-i}^T\log \tilde{z}_{-i} \right)\prod_{j\neq i}^d \tilde{z}_j^{-1},\quad \tilde{z}_{-i}=z_{-i}/z_i. 
\]
These relations imply
\begin{align}
        &\int_{A_i}\tilde\lambda(z) \mathrm{d}z \nonumber\\
				=&  \int_{(0,\infty)^d} 1_{\left\{z_i>a_i,\ z_{-i}/z_i\leq a_{-i}/a_i\right\}}  z_i^{-d-\alpha}\tilde\lambda(z/z_i)\,\mathrm{d}z\nonumber\\
				=&  \int_{(0,\infty)^d} 1_{\left\{z_i>a_i,\ \tilde{z}_{-i}\leq a_{-i}/a_i\right\}} z_i^{-d-\alpha}\exp\left(-\frac{1}{2}\log \tilde{z}_{-i}^T Q_{-i} \log \tilde{z}_{-i}+l_{-i}^T\log \tilde{z}_{-i} \right)\prod_{j\neq i}^d \tilde{z}_j^{-1}\,\mathrm{d}z\nonumber\\
				=&\int_{a_i}^\infty  \int_{[0,a_{-i}/a_i]} z_i^{-\alpha-1}\exp\left(-\frac{1}{2}\log \tilde{z}_{-i}^T Q_{-i} \log \tilde{z}_{-i}+l_{-i}^T\log \tilde{z}_{-i} \right)\left(\prod_{j\neq i}^d \tilde{z}_j^{-1}\right)\,\mathrm{d}z_i\mathrm{d}\tilde{z}_{-i}\label{eq:log-normal}\\
				=& \frac{1}{\alpha}a_i^{-\alpha}\int_{[0,a_{-i}/a_i]} z_i^{-\alpha-1}\exp\left(-\frac{1}{2}\log \tilde{z}_{-i}^T Q_{-i} \log \tilde{z}_{-i}+l_{-i}^T\log \tilde{z}_{-i} \right)\left(\prod_{j\neq i}^d \tilde{z}_j^{-1}\right)\,\mathrm{d}\tilde{z}_{-i} \nonumber
    \end{align}
		where we have used  for the third inequality the change of variable $z\to (z_i,\tilde z_{-i})$. In the last integral with respect to $\tilde z_{-i}$, we recognize a  log-normal density (up to a multiplicative factor), so that
\begin{align*}
&\int_{[0,a_{-i}/a_i]} \exp\left(-\frac{1}{2}\log \tilde{z}_{-i}^T Q_{-i} \log \tilde{z}_{-i}+l_{-i}^T\log \tilde{z}_{-i} \right)\left(\prod_{j\neq i}^d \tilde{z}_j^{-1}\right)\,\mathrm{d}\tilde{z}_{-i}\\
=& (2\pi)^{(d-1)/2}\mathrm{det}(Q_{-i})^{-1/2}\mathrm{exp}\left\{ \frac{1}{2}l_{-i}^TQ_{-i}^{-1}l_{-i}\right\} \Phi_{d-1}(\log(a_{-i}/a_i);Q_{-i}^{-1}l_{-i},Q_{-i}^{-1}).
\end{align*}
The result follows:
\begin{align*}
C_a(Q,l)&=\sum_{i=1}^d \int_{A_i}\tilde\lambda(z)\,\mathrm{d}z\\
&= (2\pi)^{(d-1)/2}\frac{1}{\alpha}\sum_{i=1}^d a_i^{-\alpha}\mathrm{det}(Q_{-i})^{-1/2}\mathrm{exp}\left\{ \frac{1}{2}l_{-i}^TQ_{-i}^{-1}l_{-i}\right\} \Phi_{d-1}(\log(a_{-i}/a_i);Q_{-i}^{-1}l_{-i},Q_{-i}^{-1}).
\end{align*}
    \end{proof}

\begin{cor}\label{3.9}
Let $Z\rightsquigarrow\mathrm{HRPar}_{a}(Q,l)$ with exponent $\alpha=-l^T1_d>0$. Then,
\begin{itemize}
\item[(i)] for all $u=(u_1,\ldots,u_d)$ such that $\sum_{i=1}^d u_i<\alpha$, we have
\[
\mathbb{E}\left[\prod_{i=1}^d Z_i^{u_i}\right]=\frac{C_a(Q,l+u)}{C_a(Q,l)}.
\]
\item[(ii)] The expectation and covariance matrix of $\log Z$ are given by
\[
\mathbb{E}\left[ \log Z_i\right]=\frac{\partial \log C_a}{\partial l_i}(Q,l)\quad i=1,\ldots,d,
\]
and
\[
\mathrm{Cov}\left(\log Z_i,\log Z_j\right)=\frac{\partial^2 \log C_a}{\partial l_i\partial l_j}(Q,l)\quad i,j=1,\ldots,d.
\]
\item[(iii)] Moreover, the expectation and covariance matrix of $\log Z$ satisfies
\[
\mathbb{E}[(\log z-\overline{\log z})(\log z-\overline{\log z})^T]=\frac{\partial \log C_a}{\partial Q}(Q,l)
\]
\end{itemize}
\end{cor}
\begin{proof}
For $\theta=(Q,l) \in \Theta$, we have for all $u=(u_1,\ldots,u_d)$ such that $\sum_{i=1}^d u_i<\alpha$
\[\theta + (0,u) \in \Theta\]
by definition of $\alpha$. Using equality \eqref{3.6.1} with $t=(0,u)$, we have
\[
\log \mathbb{E}_{\theta}\left[e^{\langle u,\log z\rangle}\right]=\log C(Q,L+u)-\log C(Q,L),
\]
taking to the exponential and developing the product implies (i).\\
The results (ii) and (iii) are straightforwards applications of \eqref{3.6.1}.
\end{proof}
\begin{ex} In dimension $d=2$ with $a=(1,1)$ and the same notations as in Example~\eqref{ex:dim-2}, we have
\[
C(Q,l)=\frac{\sqrt{2\pi}}{\alpha\sqrt{c}}\left\{ e^{l_{1}^2/2c}\Phi\left(-l_{1}/\sqrt{c}\right) +  e^{l_{2}^2/2c}\Phi\left(-l_{2}/\sqrt{c}\right)\right\}.
\]
The first order partial derivatives of $\log C$ are equal to
\begin{align*}
\frac{\partial \log C}{\partial l_1}&=-\frac{1}{l_1+l_2}+ \frac{cl_1\Phi\left(-l_{1}/\sqrt{c}\right)-\varphi(-l_1/\sqrt{c})/\sqrt{c} }{\Phi\left(-l_{1}/\sqrt{c}\right) +  e^{c(l_2^2-l_1^2)/2}\Phi\left(-l_{2}/\sqrt{c}\right)} \\
\frac{\partial \log C}{\partial l_2}&= -\frac{1}{l_1+l_2}+ \frac{cl_2\Phi\left(-l_2/\sqrt{c}\right)-\varphi(-l_2/\sqrt{c})/\sqrt{c} }{\Phi\left(-l_2/\sqrt{c}\right) +  e^{c(l_1^2-l_2^2)/2}\Phi\left(-l_1/\sqrt{c}\right)} \\
\frac{\partial \log C}{\partial c}&= -\frac{1}{2c} +\frac{1}{2}\frac{l_1\Phi(-l_1/\sqrt{c})-l_1c^{-3/2}\phi(-l_1/\sqrt{c})}{\Phi\left(-l_{1}/\sqrt{c}\right) +  e^{c(l_2^2-l_1^2)/2}\Phi\left(-l_{2}/\sqrt{c}\right)}\\
&+\frac{1}{2}\frac{l_2\Phi(-l_2/\sqrt{c})-l_2c^{-3/2}\phi(-l_2/\sqrt{c})}{\Phi\left(-l_2/\sqrt{c}\right) +  e^{c(l_1^2-l_2^2)/2}\Phi\left(-l_1/\sqrt{c}\right)}
\end{align*} 
This formulas provides respectively the expectations $\mathbb{E}[\log Z_1]$, $\mathbb{E}[\log Z_2]$ and $-\frac{1}{8}\mathbb{E}[(\log Z_1 -\log Z_2)^2]$.
Formulas for the general case $a=(a_1,a_2)$ can be deduced using Proposition \ref{scalepow}.
\end{ex}

\subsection{Simulation of HR-Pareto random vectors}
We now consider the simulation of an H\"usler-Reiss Pareto random vector $Z\rightsquigarrow\mathrm{HRPar}_{a}(Q,l)$.
Thanks to the transformation property \eqref{scalepow}, we focus on the case $a=1_d$.
In the following proposition, we denote by $S=\{x\in (0,\infty)^{d}\,:\, \|x\|_\infty=1\}$ the unit sphere and 
we use $S=\cup_{i=1}^d S_i$ with $S_i=\{x\in S\,:\, x_i=1\}$.
\begin{prop}\label{prop:simulation}
Let $Z\rightsquigarrow\mathrm{HRPar}_{1_d}(Q,l)$ with exponent $\alpha>0$. 
Then $R=\|Z\|$ and $\Theta=Z/\|Z\|$  are independent and such that
\begin{itemize}
\item[-] $R$ is a $\mathrm{Pareto}(\alpha)$-distributed real random variable, i.e. $\mathbb{P}(R>r)=r^{-\alpha}$, $r>1$; 
\item[-] $\Theta$ is a random vector on $S$ satisfying, for $i=1,\ldots,d$,
\begin{equation}\label{eq:p_i}
\mathbb{P}(\Theta\in S_i)= \frac{\mathrm{det}(Q_{-i})^{-1/2}\exp\left\{ \frac{1}{2}l_{-i}^TQ_{-i}^{-1}l_{-i}\right\} \Phi_{d-1}\left( 0;Q_{-i}^{-1}l_{-i},Q_{-i}^{-1}\right)}{\sum_{j=1}^d \mathrm{det}(Q_{-j})^{-1/2}\exp\left\{ \frac{1}{2}l_{-j}^TQ_{-j}^{-1}l_{-j}\right\} \Phi_{d-1}\left(0;Q_{-j}^{-1}l_{-j},Q_{-j}^{-1}\right)}
\end{equation}
and, given $\Theta\in S_i$, $\Theta_i=1$ and
\[
\mathcal{L}(\Theta_{-i}\mid \Theta\in S_i)=\mathcal{L}(\exp(G_i)\mid G_i\leq 0) \quad \mbox{with}\quad  G_i \rightsquigarrow \mathcal{N}_{d-1}(Q_{-i}^{-1}l_{-i},Q_{-i}^{-1}).
\] 
\end{itemize}
\end{prop}

\begin{proof} The proof is mostly a reinterpretation of the computations from the  proof of Proposition~\ref{prop:C_a(Q,l)}. The density of $Z\rightsquigarrow\mathrm{HRPar}_{1_d}(Q,l)$  is given by
\[
f_{1_d}(z;Q,l)=\frac{1}{C_{1_d}(Q,l)}1_{\{\|z\|>1\}}\tilde\lambda(z),\quad z\in (0,\infty)^d,
\]
with $\tilde\lambda$ the function defined by Equation~\eqref{eq:tildelambda}. From the proof of Proposition~\ref{prop:C_a(Q,l)}, we have 
\[
C_{1_d}(Q,l)=\sum_{i=1}^d \int_{A_i} \tilde\lambda(z)\mathrm{d}z
\]
with  
\[
A_i=\{z\in (0,\infty)^d\,:\, \|z\|>1, z/\|z\|\in S_i\},\quad i=1,\ldots,d.
\]
The expression  for $A_i$ here is slightly different but equivalent since $a=1_d$. Consequently, we get
\[
\mathbb{P}(\Theta\in S_i)=\int_{(0,\infty)^d} 1_{\{z/\|z\|\in S_i \}}f_{1_d}(z;Q,l)\mathrm{d}z=\frac{1}{C_{1_d}(Q,l)}\int_{A_i} \tilde\lambda(z)\mathrm{d}z
\]
which yields Equation~\eqref{eq:p_i} in view of Propositio~\ref{prop:C_a(Q,l)} and its proof.

\medskip
On the other hand, when $Z\in A_i$ or equivalently $\Theta\in S_i$, we have $R=\|Z\|=Z_i$ whence the change of variable $z\to (z_i,\tilde z_{-i})$ in Equation~\eqref{eq:log-normal} provides exactly the joint distribution of $(R,\Theta_{-i})$. This amounts to be the product  of  $\alpha$-Pareto and  log-normal distributions, proving the  independence of $R$ and $\Theta$ and the form of their distribution. 
\end{proof}

In order to simulate the Gaussian random variable $G_i$ conditioned on $G_i\leq 0$, we propose a recursive sampling procedure. Let $i\in\{1,\ldots,d\}$ be fixed and denote by $G_{i,j}$  the components of $G_i$. We first set $G_{i,i}=0$ and  $J=\{1,\ldots,d\}\setminus\{i\}$ the set of indices to sample. For $j\in J$, the conditional distribution of  $G_{i,j}$ given the already sampled components $G_{i,J^c}$ has a Gaussian distribution with mean and variance
 \begin{equation}\label{eq:msigma}
m_{i,j}=\Big(Q_{J,J}^{-1}\big(l_{J} -Q_{J,J^c} G_{i,J^c}\big)\Big)_j\quad \mbox{and}\quad \sigma^2_{i,j}=\Big(Q_{J,J}^{-1}\Big)_{j,j}
 \end{equation}
 subject to the constraint $G_{i,j}\leq 0$.  By the inversion method, we can sample from $G_{i,j}$ as
\[
G_{i,j}=m_{i,j}+\sigma_{i,j}\Phi^{-1}\Big(\Phi(-m_{i,j}/\sigma_{i,j})U_j\Big),\quad U_j\rightsquigarrow \mathrm{Unif}([0,1]),  
\]
where $\Phi$ denotes the standard normal cumulative distribution function.
Then, we replace $J$ by $J\setminus\{j\}$ and repeat the procedure with the next component to sample until $J$ is empty.

Note that the above computations are closely related to the distribution of extremal functions in the conditional sampling procedure of the Brown-Resnick max-stable process, see Dombry et al. \cite[section 2.2]{DER13}. Based on Proposition~\ref{prop:simulation} and the above recursive scheme, Algorithm~\ref{algo:HRPareto} describes a simulation procedure for H\"usler-Reiss Pareto random vectors.

\begin{algorithm}
\KwIn{the parameters $Q$ and $l$ of the HR-Pareto distribution}
\KwOut{a sample $Z\rightsquigarrow\mathrm{HRPar}_{1_d}(Q,l)$}
Compute $\alpha=-l^T1_d$ and sample $R\rightsquigarrow\mathrm{Pareto}(\alpha)$.

Compute $p_i=P(\Theta\in S_i)$, $i=1,\ldots,d$, according to Eq.~\eqref{eq:p_i}.

Sample $i$ from the distribution $(p_1,\ldots,p_d)$ and set $J=\{1,\ldots,d\}\setminus\{i\}$.

Initialize $G=0_d$ ($d$-dimensional null vector).

\For {$j\in J$}{
Compute $m,\sigma^2$ according to Eq.~\eqref{eq:msigma}.

Sample $U\rightsquigarrow\mathrm{Unif}([0,1])$ and set $G_{j}=m+\sigma\Phi^{-1}\Big(\Phi(-m/\sigma)U\Big)$.

Set $J=J\setminus\{j\}$.
}
Set $\Theta=\exp(G)$ and $Z=R\Theta$.

\Return{$Z$}.
\caption{Simulation of a H\"usler-Reiss Pareto random vector}\label{algo:HRPareto}
\end{algorithm}


\subsection{Maximum likelihood inference}
The exponential family property of the H\"usler-Reiss Pareto distributions makes maximum likelihood inference particularly convenient. We always suppose the threshold  $a\in(0,\infty)^d$ to be known and estimate the parameter $\theta=(Q,l)\in \Theta$ from observations  $z^{(1)},\ldots,z^{(n)} \in (0,\infty)^d\setminus[0,a]$. In the H\"usler-Reiss Pareto model, the log-likelihood of the sample writes, for $\theta=(Q,l)\in\Theta$,
    \begin{align*}
        L_n(\theta;z^{(1)},\cdots,z^{(n)})&=\frac{1}{n}\sum_{i=1}^n \log f_{a}(z^{(i)};Q,l)\\
				&= \langle  (Q,l),\overline{T}_n \rangle - \log C_{a}(Q,l) + \mathrm{cst}
    \end{align*}
    where $\overline{T}_n$ is the sufficient statistic defined by
    \[
    \overline{T}_n=\left(-\frac{1}{2n} \sum\limits_{i=1}^n(\log z^{(i)}- \overline{\log z^{(i)}})( \log z^{(i)}-\overline{\log z^{(i)}})^T ,\frac{1}{n}\sum\limits_{i=1}^n\log z^{(i)} \right)
    \]
    and the constant term $\mathrm{cst}$ does not  depend on the parameter $\theta=(Q,l)$. 
    Using the classical theory of maximum likelihood estimation for exponential families, we obtain the following result, regarding existence, uniqueness and asymptotic normality of the maximum likelihood estimator
    \[
    \hat\theta_n=\argmax_{\theta\in\Theta}L_n(\theta;z^{(1)},\cdots,z^{(n)}).
    \]

\begin{theo}\label{theo:mle} Let $a\in(0,\infty)^d$ and $n\geq 1$. \begin{enumerate}
\item[(i)] (existence and uniqueness) For observations $z^{(1)},\ldots,z^{(n)}\in [0,a]^c$, the log-likelihood $(Q,l)\mapsto L_n(Q,l;z^{(1)},\cdots,z^{(n)})$ is strictly concave on $\Theta$.
 and a maximum likelihood estimator $\hat\theta_n$ exists if and only the sample covariance matrix 
 $$
 V_n=\frac{1}{n}\sum\limits_{i=1}^n \log z^{(i)}\log z^{(i)T}-\left(\frac{1}{n}\sum\limits_{i=1}^n \log z^{(i)}\right)\left(\frac{1}{n}\sum\limits_{i=1}^n \log z^{(i)}\right)^T
 $$ is conditionally definite positive in the sense that $v^T V_n v >0$ for all $v \in \mathbb{R}^d\setminus\{0_d\}$ such that $v^T 1_d=0$.
If it exists, the maximum likelihood $\hat\theta_{n}^{mle}$ is the unique solution of the score equation
\begin{equation}\label{eq:score}
\frac{\partial \log C_{a}}{\partial \theta} (\theta)=\overline{T}_n,\quad \theta\in\Theta.
\end{equation}
\item[(ii)] (asymptotic normality) Let $\theta=(Q,l)\in\Theta$ and assume $Z^{(1)},\ldots,Z^{(n)}$ are generated from the distribution $\mathrm{HRPar}_a(Q,l)$. Then, for  $n\geq d-1$, there exists almost surely a unique maximum likelihood estimator $\hat\theta_n^{mle}$ which is asymptotically normal and efficient, that is  
     \[
     \sqrt{n}(\hat\theta_n^{mle} -\theta) \stackrel{d}\longrightarrow\mathcal{N}(0, I(\theta)^{-1}),\quad \mbox{as } n\to\infty,
     \]
where $I(\theta)$ is the Fisher Information matrix given by
\[
I(\theta)=- \frac{\partial^2 \log C_a}{\partial \theta\partial\theta^T} (\theta).
\]
\end{enumerate}
\end{theo}

\begin{rem}
In statement $i)$, if  $1_d^T\log z^{(1)},\ldots,1_d^T\log z^{(n)}$ are not all equal for $i=1,\cdots,n$ then the condition $V_n$ conditionally definite positive is equivalent to $V_n$ definite positive.
\end{rem}

The proof of Theorem~\ref{theo:mle} relies on the following Lemma.
\begin{lemm}\label{lem:conv-hull}
Recall the definition \eqref{eq:exponential-model}  of  the sufficient statistic $T(z)$. Then, the closed convex hull of the set
\[
S=\left\{T(z)\;;\ z\in (0,\infty)^d,\ z\nleq 1_d\right\}
\]
is equal to 
\[
 C=\left\{ (Q,l)\in E:Q\preceq -\frac{1}{2}(l-\overline{l})(l-\overline{l})^T \right\},
\]
where $Q_1\preceq Q_2$ means that the symmetric matrix $Q_2-Q_1$ is semi-definite positive. 
\end{lemm}
\begin{proof}[Proof of Lemma~ \ref{lem:conv-hull}]
The change of variable $u=\log z$ shows that 
\[
S=\left\{ \left(-\frac{1}{2}(u-\overline{u})(u-\overline{u})^T, u \right)\;, u\nleq 0\right\}
\]
where $\overline{u}=d^{-1}(1_d^T \log z)1_d$.
It is easily shown that $C$ is closed, convex and contains  $S$, so that  $\overline{\mathrm{conv}}(S)\subset \overline{\mathrm{conv}}(C)=C$. We consider now the reverse inclusion. Consider $U, U^{(1)},U^{(2)},\ldots$ i.i.d. with mean $l$, variance $\Sigma$ and such that $U\nleq 0$ a.s. The random element
\[
S_n=\frac{1}{n}\sum\limits_{i=1}^n \left(-\frac{1}{2}( U^{(i)}- \overline{U^{(i)}})(U^{(i)}- \overline{U^{(i)}})^T,U^{(i)}\right).
\]
belongs to $\mathrm{conv}(S)$ and, by the  law of large numbers, 
\[
S_n \stackrel{a.s.}{\longrightarrow}S_\infty= \left(-\frac{1}{2}\mathbb{E}\left((U-\overline{U})(U-\overline{U})^T\right),\mathbb{E}(U)\right),\quad n \to \infty,
\] 
so that $S_\infty\in \overline{\mathrm{conv}}(S)$. We prove below that for all $(Q,l) \in C$, one can choose $\Sigma$ such that $S_\infty=(Q,l)\in \overline{\mathrm{conv}}(S)$, proving the reverse inclusion $C\subset \overline{\mathrm{conv}}(S)$. Using $\overline{U}=d^{-1}1_d 1_d^T U$, we deduce
\begin{align*}
    \mathbb{E}\left((U-\overline{U})(U-\overline{U})^T \right)&=\mathbb{E}\left((U-d^{-1}1_d1_d^T U)(U-d^{-1}1_d1_d^T U)^T \right)\\
    &=\mathbb{E}\left(\left(I-\frac{1}{d}1_d1_d^T \right)UU^T\left(I-\frac{1}{d}1_d1_d^T \right)^T \right)\\
    &=\left(I-\frac{1}{d}1_d1_d^T\right)(\Sigma+ll^T)\left(I-\frac{1}{d}1_d1_d^T\right)^T
\end{align*}
It is proved in Lemma \ref{Projection} that the  linear operator on the space of symmetric $d\times d$ matrices  defined by
\[
P: M \mapsto \left(I-\frac{1}{d}1_d1_d^T\right)M\left(I-\frac{1}{d}1_d1_d^T\right)^T
\]
is the orthogonal projection on the linear subspace $\{M: M1_d=0\}$. 
Therefore, for $\Sigma$ such that
$P(\Sigma+ll^T)=Q$, we have
\[ S_n \longrightarrow (P(\Sigma+ll^T),l)=(Q,l).\]
In particular, since we can take $\Sigma$ among all symmetric positive semi-definite matrix, the choice $\Sigma=-2Q-P(ll^T)$ which is positive by definition of C leads to the result. Therefore $C\subset \overline{\mathrm{conv}(S)}$. 
\end{proof}

\begin{proof}[Proof of Theorem~\ref{theo:mle}]
We assume here without loss of generality that $a=1_d$. The cumulant transform $\theta\in (Q,l)\in\Theta\mapsto \log C_a(Q,l)$ is a strictly convex function. Therefore the log-likelihood $L_n$ is strictly concave as a difference of a linear function and a strictly convex function. The general theory for exponential families (see e.g. Bandorff-Nielsen \cite[Theorem 9.13]{MR3221776}) ensures that the maximum likelihood estimator exists if an only if the sufficient statistic $\overline{T}_n$ belongs to the interior of the closed convex hull of the support of $T$, that is  $\overline{T}_n \in \mathrm{int}(\overline{\mathrm{conv}(S)})=\mathrm{int}(C)$ with $S$ and $C$ defined in Lemma~\ref{lem:conv-hull}. In this case, Theorem 9.13 in Barndorff-Nielsen \cite{MR3221776} implies that the maximum likelihood estimator is unique and solves the score equation~\eqref{eq:score}. So in order to prove statement $(i)$, it remains  to prove that $ \overline{T}_n \in \mathrm{int}(\overline{\mathrm{conv}(S)})$ if and only if $V_n$ is conditionally definite positive. Note that, by Lemma~\ref{lem:conv-hull},
\begin{align*}
    \mathrm{int} (\overline{\mathrm{conv}(S)})=\mathrm{int}(C)=\left\{ (Q,l)\in E\;:\; Q\prec -\frac{1}{2}(l-\overline{l})(l-\overline{l})^\perp \text{ on vect}(1_d)^T \right\},
\end{align*}
where $Q_1\prec Q_2 \text{ on vect}(1_d)^T$ means that  $v^T(Q_2-Q_1)v>0$ for all $\mathbb{R}^{d}\setminus\{0_d\}$ such that $v^T 1_d=0$. For such $v$ and for $(Q,l)=\overline{T}_n$, we have
\begin{align*}
    &v^T\left(-Q -\frac{1}{2}(l-\overline{l})(l-\overline{l})^T\right)v\\
    =&v^T\left( \frac{1}{2n}\sum\limits_{i=1}^n(\log z^{(i)})( \log z^{(i)})^T -\frac{1}{2}\left(\frac{1}{n}\sum\limits_{i=1}^n\log z^{(i)}\right)\left(\frac{1}{n}\sum\limits_{i=1}^n\log z^{(i)}\right)^T\right)v\ge 0\\
    =& v^T V_n v
\end{align*}
whence we deduce that $\overline{T}_n\in \mathrm{int}(\overline{\mathrm{conv}(S)})$ if and only if  $V_n$ is conditionally positive. 

\medskip
Statement $ii)$ follows directly from the general theory of exponential families since the Hüsler-Reiss distributions form a full rank exponential family (see e.g. Van der Vaart \cite[Theorem 4.6]{vdV98}).

\end{proof}

\section{The generalized Hüsler-Reiss Pareto model}
       
        \subsection{Definition and transformation properties}
        \begin{defi}\label{GeneHuskerReiss}
        Let $d\geq 2$ and define $\Theta$ the  set of all $\theta=(\alpha,Q,l)$ such that:
        \begin{itemize}
            \item[-] $\alpha\in(0,\infty)^d$,
            \item[-] $Q\in \mathbb{R}^{d\times d}$ is symmetric semi-definite positive and $\mathrm{Ker} Q= \mathrm{vect}(1_d)$,
            \item[-] $l \in \mathbb{R}^d$ satisfies $l^T 1_d=-1$.
        \end{itemize} 
        For $a\in (0,\infty)^d$, the generalized H\"usler-Reiss Pareto model on $[0,a]^c=[0,\infty)^d \setminus[0,a]$ with parameters $\theta=(\alpha,Q,l)$ is defined by the density
        \begin{equation}\label{eq:GHRP}
        f_a(z;\theta)=\frac{1}{C_a(\theta)}\exp\left(-\frac{1}{2}\log z^T D_{\alpha} Q D_{\alpha} \log z + l^t D_{\alpha} \log z \right)\left(\prod\limits_{i=1}^d z_i^{-1}\right)1_{\{z \nleq a\}}
        \end{equation}
         with $C_a(\theta)$ the normalization constant and $D_\alpha$ the diagonal matrix with diagonal $\alpha$.\\
         We write $Z\rightsquigarrow\mathrm{HRPar}_{a}(\alpha,Q,l)$ for a  random vector $Z$ with density $f_{a}(z;\alpha,Q,l)$.
        \end{defi}
        For $\lambda>0$, the substitution $(\alpha,Q,l)\mapsto (\lambda \alpha,\lambda^{-1/2}Q,\lambda^{-1}l)$  leaves Equation~\eqref{eq:GHRP} invariant so that the condition $l^T 1_d=-1$ is meant to ensure that the model is identifiable. In the case  $\alpha=\bar\alpha 1_d$ with $\bar\alpha>0$,  the generalized H\"usler-Reiss model coincides with the H\"usler-Reiss Pareto model since $f_a(z;\alpha,Q,l)=f_a(z;\bar\alpha^2Q,\bar\alpha l)$ and $\bar\alpha$ is the tail index.
        
        \medskip
        Similarly as  HR-Pareto distributions, generalized HR-Pareto distributions enjoy a stability property under scale and power transformations.
        \begin{prop}\label{Genetrans}
        Let $Z \rightsquigarrow\mathrm{HRPar}_{a}(\alpha,Q,l)$.
        \begin{enumerate}
            \item [(i)] For all $u\in (0,\infty)^d$, $uZ\rightsquigarrow\mathrm{HRPar}_{ua}(\alpha,Q,l+QD_\alpha\log u)$.
            \item [(ii)] For all $\beta \in (0,\infty)^d$, $Z^{\beta}\rightsquigarrow\mathrm{HRPar}_{a^{\beta}}(\alpha/\beta,Q,l)$.
        \end{enumerate}
        \end{prop}
        \begin{proof}
        The change of variable $\tilde z=uz$ implies
        \begin{align*}
            \mathbb{P}(uZ \in A)=\int_A f_a(\tilde z/u;\alpha,Q,l)\prod_{i=1}^d u_i^{-1}\mathrm{d}\tilde z.
        \end{align*}
        Similarly as in  the proof of Proposition~\eqref{scalepow}, we check that
        \begin{align*}
            & f_a(z/u;\alpha,Q,l)\prod_{i=1}^du_i^{-1}\\
            &\quad =\frac{C_{ua}(\alpha,Q,l+Q\log u)}{\exp \left\{\frac{1}{2}\log u^T D_\alpha QD_\alpha\log u + l^TD_\alpha\log u \right\}C_a(\alpha,Q,l)}f_{ua}(z;\alpha,Q,l+QD_\alpha \log u)
         \end{align*}
         whence statement $(i)$ follows.
        The change of variable $\tilde z= z^{\beta}$ implies 
        \begin{align*}
            \mathbb{P}(Z^\beta \in A)=\int_A f_a(\tilde z^{1/\beta};\alpha,Q,l)\prod_{i=1}^d \beta_i^{-1}\tilde{z}_i^{1/\beta_i -1}\mathrm{d}\tilde z
        \end{align*}
        and simple computations result in 
        \begin{align*}
            f_a(z^{1/\beta};\alpha,Q,l)\prod_{i=1}^d \beta_i^{-1}z_i^{1/\beta_i -1}
            =\frac{C_{a^\beta}(\alpha/\beta,Q,l)}{C_a(\alpha,Q,l)\prod_{i=1}^d \beta_i}f_{a^\beta}(z;\alpha/\beta,Q,l)
        \end{align*}
        whence statement $(ii)$ follows.
        \end{proof}
        
        We deduce a simple relation between generalized HR-Pareto distribution and (standard) HR-Pareto distribution.
        \begin{cor}\label{cor:4.3}
            Let $Z\rightsquigarrow\mathrm{HRPar}_a(\alpha,Q,l)$ with $\alpha \in (0,\infty)^d$. We have $Z\stackrel{d}{=}a\tilde{Z}^{c/\alpha}$ where $\tilde{Z}\rightsquigarrow \mathrm{HRPar}_{1_d}(Q,l-QD_{\alpha}\log a)$ with exponent $c>0$. 
            \smallskip
            Moreover, we have the relationships 
            \begin{enumerate}
                \item[(i)] $C_{ua}(\alpha,Q,l+Q\log u)=\exp\left\{\frac{1}{2}\log u^T D_{\alpha} Q D_{\alpha} \log u+ l^T D_{\alpha} \log u\right\}C_a(\alpha,Q,l)$
                \item[(ii)] $C_{a^\beta}(\alpha/\beta,Q,l)=C_a(\alpha,Q,l)\prod_{i=1}^d \beta_i$.
            \end{enumerate}
        \end{cor}
        The following proposition relates the moments of generalized Hüsler-Reiss Pareto model with those of the Hüsler-Reiss Pareto model.
        \begin{prop}\label{prop:4.4} 
        Without loss of generality, assume $a=1_d$ and let $Z\rightsquigarrow \mathrm{HRPar}(\alpha,Q,l)$. Then, the expectation and the covariance matrix of $\log Z$ are given by
            \begin{align*}
                \mathbb{E}_{\alpha,Q,l}\left[\log Z_i\right]=\alpha_i^{-1}\mathbb{E}_{Q,l}\left[\log Z_i\right]\quad i=1,\cdots,d
            \end{align*}
            and 
            \begin{align*}
                \mathrm{Cov}_{\alpha,Q,l}\left(\log Z_i,\log Z_j\right)=\alpha_i^{-1}\alpha_j^{-1}\mathrm{Cov}_{Q,l}\left(\log Z_i,\log Z_j \right)
            \end{align*}
            where $\mathbb{E}_{Q,l}$ and $\mathrm{Cov}_{Q,l}$ are the expectation and covariance of H\"usler-Reiss Pareto distribution with exponent $1$.
        \end{prop}
        \begin{proof}
        Proposition \ref{Genetrans} yields $
            \mathbb{E}_{\alpha,Q,l}[\log Z_i]=\int \frac{1}{\alpha_i} \log z_i f_1(z;1_d,Q,l)\mathrm{d}z
        $. 
        Similarly, we have $\mathbb{E}_{\alpha,Q,l}[\log Z_i \log Z_j]=\frac{1}{\alpha^i \alpha^j}\mathbb{E}_{Q,l}[\log Z_i \log Z_j]$.
        Thus the result.
        \end{proof}
        \begin{rem}
        The family of the generalized Hüsler-Reiss Pareto distributions form a curved exponential family with minimal sufficient statistic $T$ given by
        \[T(z)=\left( \log z\log z^T,\log z\right).\]
        The associated natural parameter space contain positive definite matrices and the set of parameters of interest $(\alpha,Q,l)$ is included in the boundary of the natural parameter space, making the theory difficult. 
        \end{rem}
        
        \subsection{Maximum likelihood inference}
        We assume without loss of generality that  $a=1_d$ is the known threshold. Based on independent observation  $Z^{(1)}, Z^{(2)},\cdots$ with distribution  $\mathrm{HRPar}_{1_d}(\theta_0)$, $\quad\theta_0\in\theta$ we define the log-likelihood
        \[
        L_n(\theta;Z^{(1)},\cdots,Z^{(n)})=\sum_{i=1}^n \log f_{1_d}(Z^{(i)},\theta),\quad \theta\in\Theta,
        \]
        and consider maximum likelihood estimation. It should be noted that we were not able to apply directly the 'classical' maximum likelihood estimation theory from Lehman \cite{L99} that uses differentiability properties of the likelihood. Indeed, despite some substantial efforts, we could not prove the relations
        \[
        \frac{\partial^k}{\partial \theta^k}\int_{(0,\infty)^d}f_{1_d}(z;\theta)\mathrm{d}z=\int_{(0,\infty)^d}\frac{\partial^k}{\partial \theta^k}f_{1_d}(\theta,z)\mathrm{d}z,\quad k=1,2,3,
        \]
        that are required (assumption M7 in \cite[Theorem 7.5.2]{L99}). Instead, we use differentiability in quadratic mean and local expansion of the likelihood process as in van der Vaart \cite[Chapter 5]{vdV98}.
       
        \begin{prop}\label{conv:multi}
        The statistical model $\{f_{1_d}(\theta;z), \theta\in\Theta\}$ is differentiable in quadratic mean. Furthermore, the local likelihood process defined by
        \[
        \tilde{L}_n(h)=L_n\left(\theta_0+h/\sqrt{n};Z^{(1)},\cdots,Z^{(n)}\right),\quad \theta_0+h/\sqrt{n}\in\Theta,
        \]       
        satisfies,  uniformly on compact sets, 
        \begin{align}
            \tilde{L}_n(h)&=\tilde{L}_n(0)+\frac{\partial\tilde{L}_n}{\partial h}(0)^T h - \frac12 h^T I_{\theta_0} h + o_p(1),\label{eq:likelihood-expansion-1}\\
            \frac{\partial \tilde{L}_n}{\partial h}(h)&=\frac{\partial\tilde{L}_n}{\partial h}(0)-I_{\theta_0}h+o_p(1),\label{eq:likelihood-expansion-2}\\
            \frac{\partial^2 \tilde{L}_n}{\partial h \partial h^T}(h)&=-I_{\theta_0} + o_p(1) \label{eq:likelihood-expansion-3}\end{align}
            with  $I_{\theta_0}$  the Fisher information matrix at $\theta_0$.
            Furthermore,  in Equations \eqref{eq:likelihood-expansion-1}-\eqref{eq:likelihood-expansion-2}, 
            \begin{equation}\label{eq:GHRP-asynorm}
            \frac{\partial \tilde{L}_n}{\partial h}(0)=\frac{1}{\sqrt{n}}\sum\limits_{i=1}^n \frac{\partial \log f_{1_d}}{\partial \theta}(Z^{(i)},\theta_0)\rightsquigarrow \mathcal{N}(0,I_{\theta_0})
            \end{equation}
            and in  Equation \eqref{eq:likelihood-expansion-3}, the $o_P(1)$-term is even uniform on $\{\|h\|\leq n^{1/2-\varepsilon}\}$ for all $\varepsilon>0$.
            \end{prop}
        \begin{proof} Differentiability in quadratic mean is proved thanks to Lemma 7.6 in van der Vaart \cite{vdV98}. It is easily checked that $\theta\mapsto \sqrt{f_{1_d}(z;\theta)}$ is continuously differentiable for every $z$. Then we need to check that, with $\ell(\theta,z)=\log f_{1_d}(\theta,z)$, the matrix 
        \[
        I(\theta)=\mathbb{E}_\theta\left [\frac{\partial\ell}{\partial \theta}(\theta,Z)\frac{\partial\ell}{\partial \theta}(\theta,Z)^T\right]
        \]
        is well defined and continuous in $\theta$. This follows easily from the fact that the log-likelihood has the specific form
        \[
        \frac{\partial\ell}{\partial\theta}(\theta,Z)=\langle A(\theta), T(z)\rangle + B(\theta)
        \]
        with $A(\theta)$, $B(\theta)$  continuous in $\theta$ and  $T(Z)=(\log Z\log Z^T,\log Z)$. Since $T(z)$ has moment of all orders  that depend continuously of $\theta$ (this is true for the exponential family HR-Pareto and hence for the generalized HR-Pareto distributions),  $I(\theta)$ is well defined and continuous in $\theta$. From \cite[Lemma 7.6]{vdV98}, we deduce that the model is differentiable in quadratic mean. For further reference, note that by  \cite[Theorem 7.2]{vdV98}, we have
        \begin{equation}\label{eq:score-centered}
        \mathbb{E}_\theta\left[\frac{\partial\ell}{\partial\theta}(\theta,Z)\right]=0\;, \quad I(\theta)=\mathbb{E}\left[\frac{\partial\ell}{\partial\theta}(\theta,Z)\frac{\partial\ell}{\partial\theta}(\theta,Z)^T \right]
        \end{equation}
        and Equation~\eqref{eq:likelihood-expansion-1} holds for all fixed $h$ (we don't have uniformity at this point).
        
        \medskip
        We now prove the uniform asymptotic expansion \eqref{eq:likelihood-expansion-3}. The change of variable $\theta=\theta_0+h/\sqrt{n}$ yields
        \[
      \frac{\partial^2 \tilde{L}_n }{\partial h\partial h^T }(h)=\frac{1}{n}\frac{\partial^2 {L}_n }{\partial \theta \partial \theta^T }(\theta)=\frac{1}{n}\sum_{i=1}^n \frac{\partial^2 \ell}{\partial \theta \partial \theta^T}(\theta;Z^{(i)}), 
      \] 
      so that, by the law of large numbers,
           \begin{equation}\label{eq:proof-da-1}
          \frac{\partial^2 \tilde{L}_n }{\partial h\partial h^T }(0) \stackrel{a.s.}{\longrightarrow}\mathbb{E}_{\theta_0}\left[ \frac{\partial^2 \ell}{\partial \theta \partial \theta^T}(\theta_0;Z)\right]:=-J_{\theta_0}, \quad \mbox{as $n\to\infty$}.
        \end{equation}
        We don't know at this point that $J_{\theta_0}=I_{\theta_0}$, this will be proven in a final step. Thanks to Taylor-Lagrange formula,  the second-order derivative increment
      \[
      \frac{\partial^2 \tilde{L}_n }{\partial h\partial h^T }(h)-\frac{\partial^2 \tilde{L}_n }{\partial h\partial h^T }(0)
      \]
     has  norm upper bounded, for $\|h\|\leq n^{1/2-\varepsilon}$, by 
     \[
     Cn^{1/2-\varepsilon} \max_{\|h\|\leq n^{1/2-\varepsilon}}\left\| \frac{\partial^3\tilde L_n}{\partial h^3} (h)\right\| 
     =Cn^{-1-\varepsilon} \max_{\|\theta-\theta_0\|\leq n^{-\varepsilon}}\left\| \frac{\partial^3 L_n}{\partial \theta^3} (\theta)\right\|.
    \]
    The specific form 
    \[
    \ell(\theta,Z)=-\frac{1}{2}\langle \log z \log z^T, D_\alpha Q D_\alpha\rangle +\langle \log z, D_\alpha l\rangle -\log C_{1_d}(\theta)
    \]
    implies that the third order derivative is upper bounded  by 
    \[
    \left\|\frac{\partial^3 L_n}{\partial \theta^3}(\theta,z)\right\| \leq C_1+C_2\left\|\sum_{i=1}^n \log Z^{(i)} \log Z^{(i)T}\right\| 
    \]
    for some constants $C_1,C_2>0$ that does not depend on $z\in [0,1_d]^c$ and $\theta$ in a neighborhood of $\theta_0$. We deduce
    \begin{equation}\label{eq:proof-da-2}
   \left\|   \frac{\partial^2 \tilde{L}_n }{\partial h\partial h^T }(h)-\frac{\partial^2 \tilde{L}_n }{\partial h\partial h^T }(0)\right\|\leq cn^{-\varepsilon}\left(C_1+C_2\left\|\frac{1}{n}\sum_{i=1}^n \log Z^{(i)}\log Z^{(i)^T}\right\| \right).
    \end{equation}
    By the law of large number, the sample mean converge almost surely  so that the right hand side is $O_P(n^{-\varepsilon})=o_P(1)$ uniformly in $\|h\|\leq n^{1/2-\varepsilon}$. Equations~\eqref{eq:proof-da-1} and \eqref{eq:proof-da-2} together imply Equation~\eqref{eq:likelihood-expansion-3} with $J_{\theta_0}$ instead of $I_{\theta_0}$ for the moment. Equations~\eqref{eq:likelihood-expansion-2} and~\eqref{eq:likelihood-expansion-1} with $J_{\theta_0}$ instead of $I_{\theta_0}$ follow from ~\eqref{eq:likelihood-expansion-1} by integration with the $o_P(1)$ term uniform on compact set. We have already noticed that differentiability in quadratic mean implies \eqref{eq:likelihood-expansion-3} with $I_{\theta_0}$, so that necessarily the two asymptotic expansion must coincide and $J_{\theta_0}=I_{\theta_0}$. This proves Equations~\eqref{eq:likelihood-expansion-1}, \eqref{eq:likelihood-expansion-2} and \eqref{eq:likelihood-expansion-3} in their final form. Finally, in view of~\eqref{eq:score-centered}, the asymptotic normality \eqref{eq:GHRP-asynorm} is a direct consequence of the central limit Theorem.
        \end{proof}
        
        The asymptotic development of the likelihood process stated in Proposition~\ref{conv:multi} together with the Argmax Theorem allows us to study the properties of the maximum likelihood estimator (existence, consistency, asymptotic normality). An important argument is that, provided $I_{\theta_0}$ is definite positive, the asymptotic expansion of the second order differentiate \eqref{eq:likelihood-expansion-3} implies that the local likelihood process $\tilde L_n(h)$ is strictly concave on $\{\|h\|< n^{1/2-\varepsilon}\}$ with high probability. As we will see in the proof below, this entails that with high probability, the likelihood process $L_n(\theta)$ as a unique local maximizer in $\{\|\theta-\theta_0\|< n^{-\varepsilon}\}$ that we define as $\hat\theta^{mle}_n$.

     \begin{theo}\label{defmle:multi}
     Let $\theta_0\in\Theta$ with $I_{\theta_0}$  definite positive and assume the observations $Z^{(1)}, Z^{(2)},\ldots$ are independent with distribution $\mathrm{HRPar}_a(\theta_0)$. Then, there exists a  maximum likelihood estimators $\hat\theta_n^{mle}$ that is asymptotically normal and efficient, i.e. 
     \[
     \sqrt{n}(\hat\theta_n^{mle} -\theta_0) \stackrel{d}\longrightarrow\mathcal{N}(0, I_{\theta_0}^{-1})\quad \mbox{as } n\to\infty.
     \]
     \end{theo}
     \begin{proof}
         The proof relies on Proposition~\ref{conv:multi} and the Argmax theorem (van der Vaart \cite{vdV98} Corollary 5.58). Consider the stochastic processes 
        \begin{align*}
            M_n(h)=\tilde{L}_n(h)-\tilde{L}_n(0),\quad \|h|\leq n^{1/2-\varepsilon}
        \end{align*}
        and
        \begin{align*}
            M(h)=G h-\frac{1}{2}h^T I_{\theta_0}h
        \end{align*} 
        where $G$ is a centered gaussian random vector with variance $I_{\theta_0}$. 
        Proposition \ref{conv:multi} implies the convergence of $M_n$ to $M$ in distribution in $l^{\infty}(K)$ for all compact $K$. The limit process M is continuous and has an unique maximizer $h$ given by $\hat{h}=I_{\theta_0}^{-1}G\rightsquigarrow \mathcal{N}(0,I_{\theta_0}^{-1}) $. Define the maximizer
        \begin{align*}
            \hat{h}_n=\argmax_{\|h\|\leq n^{1/2-\varepsilon}} M_n(h),
        \end{align*} 
        where the argmax exists because $M_n$ is continuous on a compact set.
        The argmax theorem implies that provided $\hat{h}_n$ is tight, $\hat{h}_n\stackrel{d}{\rightarrow}\hat{h}$ as $n\rightarrow \infty$. 
        
        \medskip
        We now prove the tightness of the sequence $\hat{h}_n$, $n\geq 1$. For all $\delta>0$, there exists $R>0$ such that 
        \begin{align*}
            \mathbb{P}(\Vert \hat h\Vert \leq R)\geq 1-\delta.
        \end{align*}
        The relation
        \begin{align*}
        M(h)=M(\hat h)-\frac{1}{2}(h-\hat h)I_{\theta_0}(h-\hat h)
        \end{align*}
        implies 
        \begin{align*}
            M(\hat h)-\max_{\Vert \hat h -h\Vert\ge 1}M(h)\geq \frac12 \lambda_{\min}
        \end{align*}
        with $\lambda_{\min}>0$  the smallest eigenvalue of $I_{\theta_0}$. Therefore, with probability at least $1- \delta$, we have
        \begin{align*}
            \max_{\Vert h\Vert =R+1}M(h)\leq M(\hat h)- \frac12 \lambda_{\min}.
        \end{align*}
            The convergence in distribution of $M_n$ to $M$ in $l^{\infty}(K)$ with $K=\{h:\Vert h\Vert \le R+1\}$ implies, for large $n$,
            \begin{align}\label{4.7prof}
                \max_{\Vert h \Vert \le R}M_n(h) -\max_{\Vert h\Vert =R+1}M_n(h)\ge \frac14 \lambda_{\min}
            \end{align}
            with probability at least $1-2\delta$. The  convergence~\eqref{eq:likelihood-expansion-3} together with the positive definiteness of $I_{\theta_0}$   imply that $M_n$ is strictly concave on $\{\|h\|\leq n^{1/2-\varepsilon}\}$ with probability at least  $1-\delta$ for $n$ large. Hence, Equation~\eqref{4.7prof} implies that the maximizer $\hat h_n$ of $M_n$ belongs to  $\{\Vert h\Vert \le R+1\}$. We have proved that for large $n$,  $\mathbb{P}(\| \hat h_n\| \le R+1)\ge 1-3\delta$, establishing  the tightness of $\hat h_n$. 
            
            \medskip
            Finally, on the event $\|\hat h_n\| \le R+1$, $\hat h_n$ belongs to the interior of $\{\|h\|\leq n^{1/2-\varepsilon}$ and is  therefore  a local maximizer of $\tilde L_n$ such that $\frac{\partial \tilde L_n}{\partial h}(\hat h_n)=0$. Then $\hat \theta_n^{mle}=\theta_0+\frac{\hat h_n}{\sqrt{n}}$ is a local maximizer of $L_n$ and such that $\frac{\partial L_n}{\partial \theta}(\hat \theta_n^{mle})=0$, that is a maximum likelihood estimator. Asymptotic normality of $\hat \theta_n$ is a direct consequence of the convergence of $\hat h_n$ to $\hat h$ since $
            \sqrt{n}(\hat \theta_n -\theta_0) = \hat h_n \stackrel{d}{\rightarrow}\hat h \sim \mathcal{N}(0,I_{\theta_0}^{-1})$.
     \end{proof}
     
     \subsection{Optimizing the likelihood}
        We have proved in the previous section that, with high probability, the likelihood function $L_n$ is strictly concave on a neighborhood of $\theta_0$ of size $n^{-\varepsilon}$, $\varepsilon>0$. However, there is no reason why it should be globally convex. We discuss here to issues associated with the likelihood optimization. The first is the initialization of an optimization algorithm and will be addressed thanks to a simple moment estimator that is $\sqrt{n}$-consistent and can serve as a starting point of optimization routines. The second point is how we can take advantage of the biconcavity of the problem: although not globally concave, the log-likelihood is biconcave in the sense that both partial applications $\alpha\mapsto L_n(\alpha,Q,l)$ and $(Q,l)\mapsto L_n(\alpha,Q,l)$ are concave. In this context, it is natural to consider alternate convex optimization. 
        
        \begin{prop}\label{firstestimatorprop}
         Let $\theta=(\alpha,Q,l) \in \Theta$ and assume the observations $Z^{(1)},Z^{(2)} \cdots$ independent with distribution $\mathrm{HRPar}_{1_d}(\theta)$. For $j=1,\ldots,d$ define 
         \[
         N_{n,j}=\frac{1}{n}\sum\limits_{i=1}^n  1_{\{Z^{(i)}_j>1\}}\quad \mbox{and}\quad  O_{n,j}=\frac{1}{n}\sum\limits_{i=1}^n1_{\{Z^{(i)}_j>1\}} \log Z_j^{(i)}.
         \]
         Then the estimator $\hat \theta_0 =(\hat \alpha_0 , \hat Q_0 , \hat l_0)$ defined by
         \[
         \hat \alpha_0=(N_{n,j}/O_{n,j})_{1\leq j\leq d}\quad \mbox{and}\quad 
(\hat Q_0, \hat l_0)=\argmax_{Q,l} L_n(\hat\alpha_0,Q,l)
         \]
         is strongly consistent and asymptotically normal.
        \end{prop}
        \begin{proof}
        For $j=1,\ldots, d$, the thresholded marginal $Z_j|Z_j>1$ are distributed according to a Pareto distribution with parameter $\alpha_j$ , so that $\mathbb{E}_{\theta}[\log Z_j\mid Z_j>1]=\alpha_j^{-1}$. Hence, by the law of large numbers
        \[
        \frac{N_{n,j}}{O_{n,j}}\stackrel{a.s}\longrightarrow \frac{\mathbb{P}_\theta(Z_j>1)}{\mathbb{E}_\theta( 1_{\{Z_j>1\}}\log Z_j)}=\Big( \mathbb{E}_{\theta}[\log Z_j\mid Z_j>1]\Big)^{-1}=\alpha_j
        \]
        so that $\hat\alpha_0$ is a consistent estimator for $\alpha$. 
        
        On the other hand, the vector $Z^{\alpha}$ has H\"usler-Reiss distribution $\mathrm{HRP}(Q,l)$, so that Theorem~\ref{theo:mle} suggests the maximum-likelihood estimator
        \[
        (\hat Q,\hat l)=\argmax_{Q,l}L_n(\alpha,Q,l)=\Psi\Big(\overline T_n(D_\alpha \log Z^{(1)},\ldots,D_\alpha \log Z^{(n)})\Big),
        \]
        where $\Psi(\bar t)$ denotes the unique solution of the score equation $ \frac{\partial \log C}{\partial\theta}(Q,l)=\bar t$. As a general result for full exponential families (see e.g. Bandorff-Nielsen \cite{BN17}), $\Psi$ is a diffeomorphism. Since $\alpha$ is unknown and estimated by $\hat\alpha_0$, we set rather
        \[
        \hat\theta_0=(\hat Q_0,\hat l_0)=\Psi\Big(\overline T_n(D_{\hat\alpha_0} \log Z^{(1)},\ldots,D_{\hat\alpha_0} \log Z^{(n)})\Big).
        \]
        Some simple computations show
        \[
        \overline T_n(D_{\hat\alpha_0} \log Z^{(1)},\ldots,D_{\hat\alpha_0} \log Z^{(n)})= D_{N_n/O_n}
        \]
        where $N_n$, $O_n$ and $N_n/O_n$ denotes the vectors with components $N_{n,j}$, $O_{n,j}$ and $N_{n,j}/O_{n,j}$ respectively, and  
        \[
        M_n=\frac{1}{n}\sum\limits_{i=1}^d \log Z^{(i)}\quad\mbox{and}\quad 
        V_n=\frac{1}{n}\sum\limits_{i=1}^d \log Z^{(i)} (\log Z^{(i)})^T.
        \]
        Hence $\hat\theta_0$ can be written in the form  
        \[
        \hat\theta_0=\Theta(N_n,O_n,M_n,V_n)
        \]
        with differentiable function $\Theta$. The law of large number ensures the almost sure convergence of $N_n,O_n,M_n,V_n$ as $n\to\infty$, whence strong consistency  $\hat\theta_0 \stackrel{ a.s.}\to \theta$ follows. The central limit theorem ensures the asymptotic normality of $(N_n,O_n,M_n,V_n)$, whence the asymptotic normality of $\hat\theta_0$ is deduced via the  $\delta$-method (van der Vaart \cite[Theorem 3.1]{vdV98}).
        \end{proof}

        \begin{theo}
        Let $\theta_0=(\alpha_0,Q_0,l_0)\in \Theta$ and assume the observations $Z^{(1)},Z^{(2)},\cdots$ independent with distributions $\mathrm{HRPar}(\theta_0)$. 
        Define $\hat\theta_0$ as in proposition \eqref{firstestimatorprop} and $$V_n=\left\{ \theta \in\Theta: \Vert \theta-\theta_0\Vert< n^{1/2-\varepsilon}\right\}.$$         Define $\hat\theta_n^{mle}$ as the unique minimizer of the negative log-likelihood on $V_n$, i.e.
        \begin{align*}
            \hat\theta_n^{mle}=\argmin_{\theta \in V_n}\quad -L_n(\theta;Z^{(1)},\cdots,Z^{(n)}).
        \end{align*}
        Consider the alternating minimization estimators $\hat \theta^{(i)}=(\hat \alpha^{(i)},\hat Q^{(i)},\hat l^{(i)})$ defined by the recursive algorithm
        \begin{align}
            \left\{\begin{array}{lll}
                 \hat\alpha^{(i+1)}&=\argmin_{\alpha} &-L_n(\alpha,\hat Q^{(i)},\hat l^{(i)}; Z^{(1)},\cdots,Z^{(n)}) \\
                 (\hat Q^{(i+1)},\hat l^{(i+1)})&=\argmin_{Q,l} &-L_n(\hat \alpha^{(i+1)},Q,l; Z^{(1)},\cdots,Z^{(n)})
            \end{array}
            \right.\quad \text{for }i>0
        \end{align}
        and initialized with $\hat\theta^{(0)}=\hat\theta_0$.
        Then, with high probability, the sequence of estimators $(\hat\theta^{(i)})_{i\geq 0}$ converges almost surely to $\hat\theta_n^{mle}$, i.e.
        \begin{align}
            \mathbb{P}\left(\lim_{i \rightarrow \infty}\hat\theta^{(i)}=\hat\theta_n^{mle}\right)\rightarrow 1, \quad\text{as }n\rightarrow\infty.
        \end{align}
        \end{theo}
        \begin{proof}
      The starting point estimator writes
      \begin{align*}
          \hat\theta_0=\theta_0+\frac{1}{\sqrt{n}}(\sqrt{n}(\hat\theta_0-\theta_0)).
      \end{align*}
      Proposition \ref{firstestimatorprop} and Prohorov's theorem implies that $\hat\theta_0 \in V_n$ with high probability. Assuming the log-likelihood strictly concave on $V_n$, we show by recurrence that each iterate of the alternating minization algorithm belongs to $V_n$.
      Define the level set  
      \[
      \mathcal{L}_i=\{\theta: -L_n(\theta)\leq -L_n(\hat\theta^{(i)})+\delta\},\quad i\geq 0,
      \]
      where $\delta>0$ is such that  $\mathcal{L}_i\cap \partial V_n=\varnothing$.  By convex optimization theory, the intersection between $\mathcal{L}_i$ and $V_n$ is a convex set. Let $B_1$ and $B_2$ be open balls centered at $\hat\theta^{(i)}$ and $(\hat\alpha^{(i+1)},\hat Q^{(i)},\hat l^{(i)})$ such that $B_1$ and $B_2$ are subset of $\mathcal{L}_i$. The biconvex property of $-L_n$ implies that  the convex hull $\mathrm{conv}(B_1,B_2)$  is a subset of $\mathcal{L}_i$. It results that $\mathrm{conv}(B_1,B_2)\subset \mathcal{L}_i \cap V_n$. A similar reasoning concludes that $\hat\theta^{(i)}\in V_n$ for all  $i\geq 0$ and therefore the alternating minimization estimators $ \hat\theta^{(i)}$ converge to the unique minimizer in $V_n$. 
        \end{proof}

        \subsection{A likelihood ratio test for  $\alpha_1=\cdots=\alpha_d$}
        Following the development of generalized pareto models, a natural question that arise when one is given a i.i.d. sample $Z^{(1)},\ldots ,Z^{(n)}$ with distribution $\mathrm{HRPar}(\alpha,Q,l)$ is whether the Pareto model would be enough to modelize the data.  The following theorem provides a likelihood ratio test for testing $\alpha_1=\cdots=\alpha_d$. 
        
        \begin{theo}
        Let $\theta_0=(\alpha,Q,l) \in \Theta$ with $\alpha=(\alpha_1,\ldots,\alpha_d)$.
        Let $Z^{(1)},\ldots, Z^{(n)}$ be i.i.d. with distribution $\mathrm{HRPar}(\theta_0)$. Denote by $\hat{\theta}_n$ the maximum likelihood estimator in the Generalised Hüssler-Reiss Pareto model and $\hat\theta_0$ the maximum likelihood estimation in the Hüssler-Reiss Pareto model and define the likelihood log-ratio by
        \begin{align*}
            \Delta_n=L_n(\hat \theta_n)-L_n(\hat\theta_0).
        \end{align*}
        Then, under the null hypothesis $\alpha_1=\cdots=\alpha_d$, the distribution of $2\Delta_n$ converge to a chi-squared distribution with $d-1$ degree of freedom, i.e.
        \begin{align*}
            2 (L_n(\hat \theta_n)-L_n(\hat\theta_0)) \stackrel{d}{\rightarrow} \chi^2(d-1).
        \end{align*}
        \end{theo}
        \begin{proof}
            Denote by $\Theta_0$ the subset of $\Theta$ defined by $\left\{ (\alpha,Q,l)\in \Theta : \alpha_1=\cdots,\cdots= \alpha_d \right\}$. Consider the local log-likelihood process $\tilde{L}_n$ and its maximiser $\hat{h}_n$ on $\Theta$. Likewise, denote by $\hat{h}^{0}_n$ the maximiser of $\tilde{L}_n$ on $\Theta_0$. 
            We prove below  that $2(\tilde{L}_n(\hat{h}_n)-\tilde{L}_n(\hat{h}_n^0)) \stackrel{d}{\rightarrow}\chi^2(p-1)$.
            Simple calculations imply that the Taylor expansion of $\tilde{L}_n$ at $\hat{h}_n$ writes 
            \begin{align*}
                \tilde{L}_n(h)=\tilde{L}_n(\hat{h}_n)-\frac{1}{2}(h-\hat{h}_n)I_{\theta_0}(h-\hat h_n)+o_p(1)
            \end{align*}
            where the $o_p$ term is uniform on compact sets containing $\hat{h}_n$. Taking a compact $K$ large enough to contain both $\hat{h}_n$ and $\hat{h}_n^0$, we have
            \begin{align*}
                2\left(\tilde{L}_n(\hat{h}_n)-\tilde{L}_n(\hat{h}_n^0)\right)&=\min_{h\in K \cap \Theta_0} 2\left(\tilde{L}_n(\hat{h}_n)-\tilde{L}_n(h)\right)\\&=\min_{h\in K \cap \Theta_0} (h-\hat{h} _n)I_{\theta_0}(h-\hat{h}_n)+o_p(1)
            \end{align*}
            Defining $\langle \cdot,\cdot \rangle_{I_{\theta_0}}$ as the inner product induced by $I_{\theta_0}$, i.e. $\langle a,b \rangle_{I_{\theta_0}}=a^t I_{\theta_0} b$, we get
            \begin{align*}
                2\left(\tilde{L}_n(\hat{h}_n)-\tilde{L}_n(\hat{h}_n^0)\right)&=\min_{h\in K \cap \Theta_0}\Vert h-\hat{h}_n\Vert_{I_{\theta_0}}^2 +o_p(1)
            \end{align*}
            The minimum is reached for $h$ the orthogonal projection of $\hat{h}_n$ into $\Theta_0$ for the $\Vert.\Vert_{I_{\theta_0}}$ norm. Thus, we have
            \begin{align*}
            2\left(\tilde{L}_n(\hat{h}_n)-\tilde{L}_n(\hat{h}_n^0)\right)&=\sum_{i=1}^{p-1}\langle \hat{h}_n, e_i\rangle^2_{I_{\theta_0}} +o_p(1)
            \end{align*}
            where $(e_1,\cdots,e_{p-1})$ is an orthonormal basis of $\Theta_0^{\perp}$. Theorem \eqref{defmle:multi}  implies $$\left(\langle\hat{h}_n,e_i\rangle_{I_{\theta_0}}\right)_{1\le i\le p-1} \stackrel{d}{\rightarrow} \mathcal{N}(0_{p-1},I_{p-1})$$
            which in turn results in $$2\left(\tilde{L}_n(\hat{h}_n)-\tilde{L}_n(\hat{h}_n^0)\right)\stackrel{d}{\rightarrow}\chi^2(p-1).$$
        \end{proof}
       
        \begin{appendix}
        \begin{lemm}\label{Projection}
            Let $\mathcal{S}_d$ denote the linear space of symmetric $d\times d$ matrices and $P$ the linear operator defined as 
            $$
            P:\mathcal{S}_d\rightarrow \mathcal{S}_d,\quad Q \mapsto (I-\frac{1}{d}1_d1_d^T)Q(I-\frac{1}{d}1_d1_d^T).
            $$
            Then $P$ is the orthogonal projection on the linear subspace $\mathcal{S}_d^0=\left\{S \in E: S1_d=0\right\}$.
        \end{lemm}
        \begin{proof}
        Let $S \in E$, we have 
        \begin{align*}
            P^2(S)&=\left(I-\frac{1}{d}1_d1_d^T\right)^2 S \left(I-\frac{1}{d}1_d1_d^T\right)^2\\
            &=\left(I-\frac{2}{d}1_d1_d^T+\frac{1}{d^2}1_d1_d^T1_d1_d^T \right)S \left(I-\frac{2}{d}1_d1_d^T+\frac{1}{d^2}1_d1_d^T1_d1_d^T \right)\\
            &=\left(I-\frac{1}{d}1_d1_d^T\right) S \left(I-\frac{1}{d}1_d1_d^T\right)\\&=P(S).
        \end{align*}
        Therefore $P$ is idempotent.\\
        For $S \in \mathcal{S}_d^0$, we have
        \begin{align*}
            P(S)&=\left(I-\frac{1}{d}1_d1_d^T\right) S \left(I-\frac{1}{d}1_d1_d^T\right)\\
            &=S-\frac{2}{d}S1_d1_d^T + \frac{1}{d^2}1_d1_d^TS1_d1_d^T\\
            &=S.
        \end{align*}
        Therefore $P$ acts  as the identity on $\mathcal{S}_d$.\\
        For $S \in (\mathcal{S}_d^0)^\perp$, we have 
        \begin{align*}
            P(S)&=...=0
        \end{align*}
        Therefore $P$ is null on $(\mathcal{S}_d^0)^\perp$. This concludes the proof.
        \end{proof}
        
        \end{appendix}

\providecommand{\bysame}{\leavevmode\hbox to3em{\hrulefill}\thinspace}
\providecommand{\MR}{\relax\ifhmode\unskip\space\fi MR }
\providecommand{\MRhref}[2]{%
  \href{http://www.ams.org/mathscinet-getitem?mr=#1}{#2}
}
\providecommand{\href}[2]{#2}


\begin{thebibliography}{10}

\bibitem{BdH74}
A.~A. Balkema and L.~de~Haan, \emph{Residual life time at great age}, Ann.
  Probability \textbf{2} (1974), 792--804. \MR{0359049}

\bibitem{MR3221776}
O.~Barndorff-Nielsen, \emph{Information and exponential families in statistical
  theory}, Wiley Series in Probability and Statistics, John Wiley \& Sons,
  Ltd., Chichester, 2014, Reprint of the 1978 original [MR0489333].
  \MR{3221776}

\bibitem{BN17}
L.~R. Belzile and J.~G. Ne\v~slehov\'a, \emph{Extremal attractors of liouville
  copulas}, arXiv:1704.03377, 2017.

\bibitem{breiman:1965}
L.~Breiman, \emph{On some limit theorems similar to the arc-sin law}, Theory of
  Probability and its Applications \textbf{10} (1965), no.~2, 323--331.

\bibitem{Coles01}
S.~Coles, \emph{An introduction to statistical modeling of extreme values},
  Springer Series in Statistics, Springer-Verlag London, Ltd., London, 2001.
  \MR{1932132}

\bibitem{CT91}
S.~G. Coles and J.~A. Tawn, \emph{Modelling extreme multivariate events}, J.
  Roy. Statist. Soc. Ser. B \textbf{53} (1991), no.~2, 377--392. \MR{1108334}

\bibitem{DM08}
R.~A. Davis and T.~Mikosch, \emph{Extreme value theory for space-time processes
  with heavy-tailed distributions}, Stochastic Process. Appl. \textbf{118}
  (2008), no.~4, 560--584. \MR{2394763}

\bibitem{dHR77}
L.~de~Haan and S.~I. Resnick, \emph{Limit theory for multivariate sample
  extremes}, Z. Wahrscheinlichkeitstheorie und Verw. Gebiete \textbf{40}
  (1977), no.~4, 317--337. \MR{0478290}

\bibitem{D78}
P.~Deheuvels, \emph{Caract\'erisation compl\^ete des lois extr\^emes
  multivari\'ees et de la convergence aux types extr\^emes}, Publ. Inst.
  Statist. Univ. Paris \textbf{23} (1978), 1--36.

\bibitem{DER13}
C.~Dombry, F.~{\'E}yi-Minko, and M.~Ribatet, \emph{Conditional simulation of
  max-stable processes}, Biometrika \textbf{100} (2013), no.~1, 111--124.
  \MR{3034327}

\bibitem{G43}
B.~Gnedenko, \emph{Sur la distribution limite du terme maximum d'une s\'erie
  al\'eatoire}, Ann. of Math. (2) \textbf{44} (1943), 423--453. \MR{0008655}

\bibitem{GS10}
G.~Gudendorf and J.~Segers, \emph{Extreme-value copulas}, Copula theory and its
  applications, Lect. Notes Stat. Proc., vol. 198, Springer, Heidelberg, 2010,
  pp.~127--145. \MR{3051266}

\bibitem{hult:lindskog:2006}
H.~Hult and F.~Lindskog, \emph{Regular variation for measures on metric
  spaces}, Publ. Inst. Math. (Beograd) (N.S.) \textbf{80(94)} (2006), 121--140.
  \MR{2281910 (2008g:28016)}

\bibitem{HD13}
D.~Huser and Davison~A. C., \emph{Composite likelihood estimation for the
  brown–resnick process}, Biometrika \textbf{100} (2013), no.~1, 511--518.

\bibitem{HR89}
J.~H\"usler and R.-D. Reiss, \emph{Maxima of normal random vectors: between
  independence and complete dependence}, Statist. Probab. Lett. \textbf{7}
  (1989), no.~4, 283--286. \MR{980699}

\bibitem{J15}
H.~Joe, \emph{Dependence modeling with copulas}, Monographs on Statistics and
  Applied Probability, vol. 134, CRC Press, Boca Raton, FL, 2015. \MR{3328438}

\bibitem{KRSW16}
A.~Kiriliouk, H.~Rootz\'en, J.~Segers, and J.~Wadsorth, \emph{Peaks over
  thresholds modelling with multivariate generalized pareto distributions},
  arXiv:1612.01773, 2016.

\bibitem{L99}
E.~L. Lehmann, \emph{Elements of large-sample theory}, Springer Texts in
  Statistics, Springer-Verlag, New York, 1999. \MR{1663158}

\bibitem{NJL09}
A.~K. Nikoloulopoulos, H.~Joe, and H.~Li, \emph{Extreme vale properties of
  multivariate t copulas}, Extremes \textbf{12} (2009), no.~2, 129--148.

\bibitem{KJLG17}
Krupskii P., Joe H., Lee D., and Genton~M. G., \emph{Extreme value limit of
  convolution of exponential and multivariate normal distribution - link to
  h\"usler-reiss distribution}, arXiv:??, 2017.

\bibitem{R87}
S.~I. Resnick, \emph{Extreme values, regular variation, and point processes},
  Applied Probability. A Series of the Applied Probability Trust, vol.~4,
  Springer-Verlag, New York, 1987. \MR{900810}

\bibitem{R07}
\bysame, \emph{Heavy-tail phenomena}, Springer Series in Operations Research
  and Financial Engineering, Springer, New York, 2007, Probabilistic and
  statistical modeling. \MR{2271424}

\bibitem{RSW17}
H.~Rootz\'en, J.~Segers, and J.~Wadsorth, \emph{Multivariate peaks over
  thresholds models}, Extremes (2017), Forthcoming.

\bibitem{RT06}
H.~Rootz\'en and N.~Tajvidi, \emph{Multivariate generalized {P}areto
  distributions}, Bernoulli \textbf{12} (2006), no.~5, 917--930. \MR{2265668}

\bibitem{S60}
M.~Sibuya, \emph{Bivariate extreme statistics. {I}}, Ann. Inst. Statist. Math.
  Tokyo \textbf{11} (1960), 195--210. \MR{0115241}

\bibitem{TdO58}
J.~Tiago~de Oliveira, \emph{Extremal distributions}, Rev. Fac. Sci. Lisboa,
  Ser. A 7 (1958), 215--227.

\bibitem{vdV98}
A.~W. van~der Vaart, \emph{Asymptotic statistics}, Cambridge Series in
  Statistical and Probabilistic Mathematics, vol.~3, Cambridge University
  Press, Cambridge, 1998. \MR{1652247 (2000c:62003)}

\bibitem{WT14}
J.~L. Wadsworth and J.~A. Tawn, \emph{Efficient inference for spatial extreme
  value processes associated to log-{G}aussian random functions}, Biometrika
  \textbf{101} (2014), no.~1, 1--15. \MR{3180654}

\end{thebibliography}
\end{document}